\documentclass[a4paper,fleqn]{cas-dc}

\usepackage[numbers]{natbib}
\usepackage{subfigure}
\usepackage{algorithm,algorithmic}

\def\tsc#1{\csdef{#1}{\textsc{\lowercase{#1}}\xspace}}
\tsc{WGM}
\tsc{QE}
\tsc{EP}
\tsc{PMS}
\tsc{BEC}
\tsc{DE}

\newproof{proof}{Proof}
\newtheorem{definition}{Definition}
\newtheorem{example}{Example}
\newtheorem{property}{Property}
\newtheorem{corollary}{Corollary}

\begin{document}
\let\WriteBookmarks\relax
\def\floatpagepagefraction{1}
\def\textpagefraction{.001}

\shorttitle{Publishing Microdata through Mutual Cover}

\shortauthors{Li et~al.}

\title [mode = title]{MuCo: Publishing Microdata with Privacy Preservation through Mutual Cover}

\author[1]{Boyu Li}[]
\ead{afterslby@outlook.com}
\credit{Conceptualization of this study, Methodology, Software}

\affiliation[1]{organization={School of Information Communication, National University of Defense Technology},
	city={Wuhan},
	postcode={430000}, 
	tate={Hubei},
	country={China}}

\author[1]{Jianfeng Ma}[]
\ead{majianfeng@189.cn}
\credit{Methodology, Data curation, Software}

\author[1]{Junhua Xi}[]
\ead{hjh17@nudt.edu.cn}
\credit{Methodology, Data curation, Writing - Original draft preparation}

\author[1]{Lili Zhang}
\ead{l1lzhang@126.com}
\credit{Methodology}

\author[1]{Tao Xie}
\ead{xietao09@nudt.edu.cn}
\credit{Methodology}

\author[1]{Tongfei Shang}
\cormark[1]
\credit{Conceptualization of this study, Methodology}
\ead{340445698@qq.com}

\cortext[cor1]{Corresponding author}

\begin{abstract}
We study the anonymization technique of $k$-anonymity family for preserving privacy in the publication of microdata. Although existing approaches based on generalization can provide good enough protections, the generalized table always suffers from considerable information loss, mainly because the distributions of QI (Quasi-Identifier) values are barely preserved and the results of query statements are groups rather than specific tuples. To this end, we propose a novel technique, called the Mutual Cover (MuCo), to prevent the adversary from matching the combination of QI values in published microdata. The rationale is to replace some original QI values with random values according to random output tables, making similar tuples to cover for each other with the minimum cost. As a result, MuCo can prevent both identity disclosure and attribute disclosure while retaining the information utility more effectively than generalization. The effectiveness of MuCo is verified with extensive experiments.
\end{abstract}

\begin{highlights}
\item We propose MuCo that supports publishing microdata with privacy preservation.
\item MuCo preserves more information utility than generalization while achieving great protection performance.
\item The anonymization process of MuCo is hidden for the adversary.
\item MuCo provides impressive privacy protection, little information loss, and accurate query answering.
\end{highlights}


\maketitle

\section{Introduction}
In recent years, the massive digital information of individuals has been collected by numerous organizations. The data holders, also known as curators, use the data for data mining tasks, meanwhile they also exchange or publish microdata for further comprehensive research. However, the publication of microdata poses critical threats to the privacy of individuals \cite{agrawal2000privacy, agrawal2005a, DENHAM2020113380, liu2015result, LI2023393}. Consequently, many anonymization techniques are proposed to defend against various privacy disclosures while maintaining the utility of individual information as much as possible \cite{fung2010privacy, zakerzadeh2014towards, MEHTA20221423}.

Typically, the attributes in microdata can be divided into three categories: (1) Explicit-Identifier (EI, also known as Personally-Identifiable Information), such as name and social security number, which can uniquely or mostly identify the record owner; (2) Quasi-Identifier (QI), such as age, gender and zip code, which can be used to re-identify the record owner when taken together; and (3) Sensitive Attribute (SA), such as salary and disease, which contains the confidential information of individuals. According to the work of Sweeney \cite{sweeney2002k}, even with all EI attributes being removed, the record owners can still be re-identified by matching the combination of QI values.

Generalization \cite{fung2010privacy, MORTAZAVI2020113454} is one of the most widely used privacy-preserving techniques. It transforms the values on QI attributes into general forms, and the tuples with equally generalized values constitute an equivalence group. In this way, records in the same equivalence group are indistinguishable. $k$-Anonymity \cite{sweeney2002k, QIAN2024122343} ensures that the probability of identity disclosure is at most $1/k$. For instance, Figure \ref{fig_1b} is a generalized table of Figure \ref{fig_1a} that complies with 2-anonymity, and the adversary has to acquire at least two different tuples by matching the age value of any person.

\begin{figure}[t]
	\vspace{0.1in}
	\centering
	\subfigure[Original table]{
		\begin{minipage}[t]{0.30\linewidth}
			\centering
			\includegraphics[width=1.in]{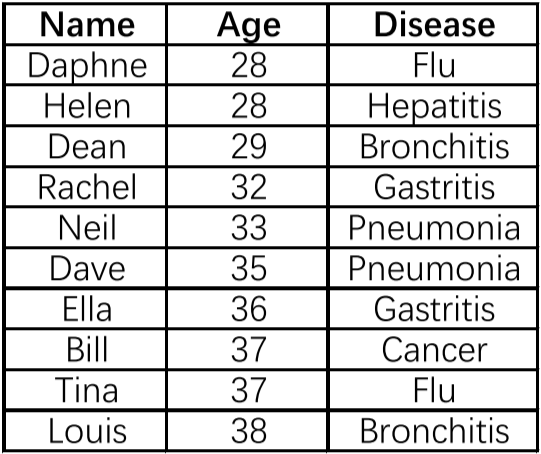}
			\vspace{1pt}
			\label{fig_1a}
	\end{minipage}}
	\hspace{1pt}
	\subfigure[2-anonymity]{
		\begin{minipage}[t]{0.30\linewidth}
			\centering
			\includegraphics[width=1.in]{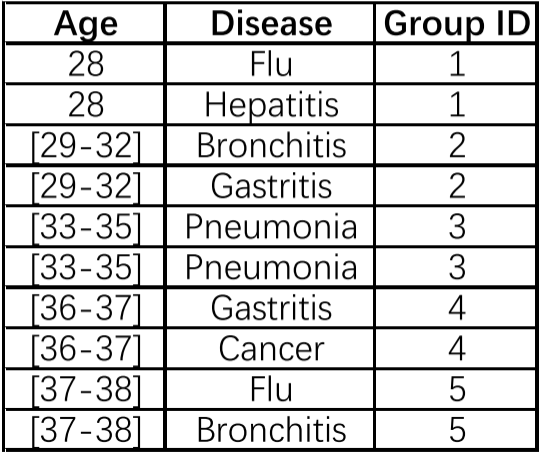}
			\label{fig_1b}
	\end{minipage}}
	\hspace{0.3pt}
	\subfigure[5-diversity]{
		\begin{minipage}[t]{0.3\linewidth}
			\centering
			\includegraphics[width=1.in]{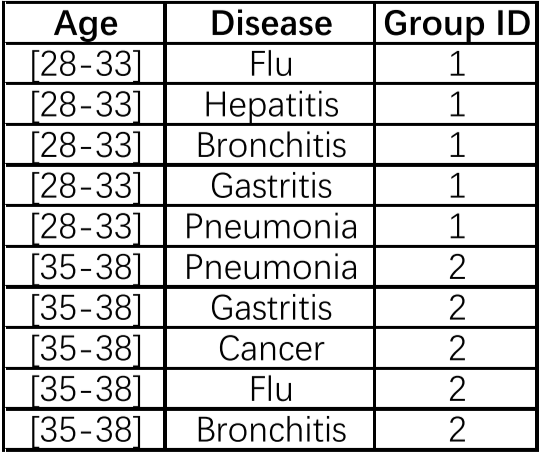}
			\label{fig_1c}
	\end{minipage}}
	\caption{An example of generalization.}
	\label{fig_1}
\end{figure}

\subsection{Motivation}
Although the generalization for $k$-anonymity provides enough protection for identities, it is vulnerable to the attribute disclosure \cite{machanavajjhala2006l}. For instance, in Figure \ref{fig_1b}, the sensitive values in the third equivalence group are both ``pneumonia''. Therefore, an adversary can infer the disease value of Dave by matching his age without re-identifying his exact record. To prevent such disclosure, many effective principles have been proposed, such as $l$-diversity \cite{machanavajjhala2006l} and $t$-closeness \cite{li2007t}. For example, Figure \ref{fig_1c} is the generalized version of Figure \ref{fig_1a} complying with 5-diversity, such that the proportion of each sensitive value inside the equivalence group is no more than $1/5$. Thus, for any individual, the adversary has to obtain at least five different sensitive values by matching the age value.

However, despite protecting against both identity disclosure and attribute disclosure, the information loss of generalized table cannot be ignored. On the one hand, the generalized values are determined by only the maximum and the minimum QI values in equivalence groups, causing that the equivalence groups only preserve the ranges of QI values and the number of records. Consequently, the distributions of QI values are hardly maintained and the information utility is reduced significantly. For instance, as shown in Figure \ref{fig_2}, the red polyline and the magenta polyline represent the distributions on age in Figure \ref{fig_1a} and Figure \ref{fig_1c}, respectively. We can observe that the original distribution is barely preserved in the generalized table. On the other hand, the partition of equivalence groups also increases the information loss of anonymized table because the results of query statements are always the matching equivalence groups rather than the specific matching tuples. For example, if we want to select the tuples whose age values are more than 30 in Figure \ref{fig_1c}, both equivalence groups are considered as the results. 

\begin{figure}[!t]
	\centering
	\includegraphics[height=1.5in]{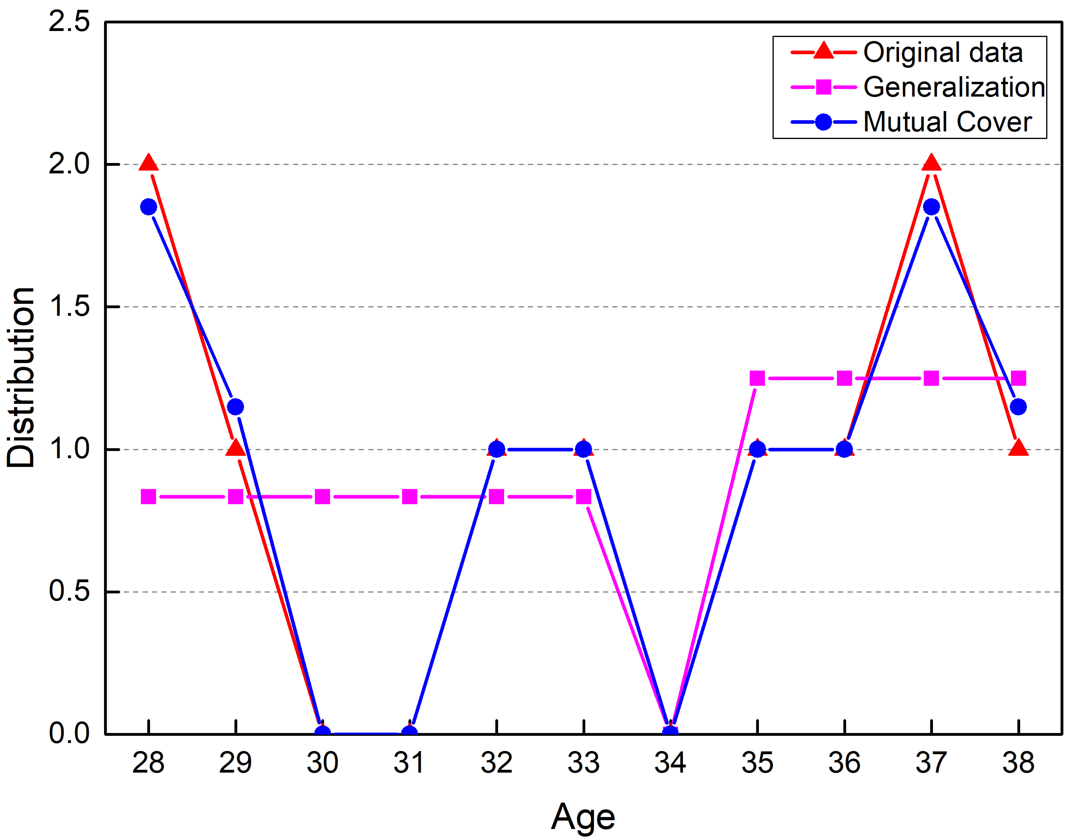}
	\caption{The distribution of age values.}
	\label{fig_2}
\end{figure}

Moreover, the level of protection for identities is possible to be compelled to increase for meeting the condition of $l$-diversity. For example, the generalized table in Figure \ref{fig_1c} must comply with at least 5-anonymity for satisfying 5-diversity even if the demand for protecting identities is not that high. The over-protection causes more information loss because a larger equivalence group often needs larger ranges of generalized values to cover the QI values. For instance, the generalized table in Figure \ref{fig_1b} maintains more detailed QI values than that in Figure \ref{fig_1c}. As a result, to decrease information loss, the protection for sensitive values should not affect the protection for identities excessively.

\subsection{Contributions}
In this work, we propose a novel technique called the Mutual Cover (MuCo) to impede adversary from matching the combination of QI values while overcoming the above issues. The key idea of MuCo is to make similar tuples to cover for each other by randomizing their QI values according to random output tables. 

Specifically, there are three main steps in the proposed approach. First, MuCo partitions the tuples into groups and assigns similar records into the same group as far as possible. Second, the random output tables, which control the distribution of random output values within each group, are calculated to make similar tuples to cover for each other at the minimal cost. Finally, MuCo generates anonymized microdata by replacing the original QI values with random values according to the random output tables. For instance, for the original table in Figure \ref{fig_1a}, MuCo partitions the records into four groups and calculates random output tables on age as shown in Figure \ref{fig_3}. In the random output tables, the rows correspond to the records, and the columns correspond to the ranges of age values. Every entry value denotes the probability that the record carries the column value in the anonymized table. For example, we can observe that Helen is covered with Daphne and Dean, and her age outputs 28 with a probability of 0.7129 and outputs 29 with a probability of 0.2871. Then, MuCo generates an anonymized table in which the original QI values are replaced by the random values according to the random output tables.

\begin{figure}[t]
	\centering
	\subfigure[]{
		\begin{minipage}[t]{0.45\linewidth}
			\centering
			\includegraphics[width=1.3in]{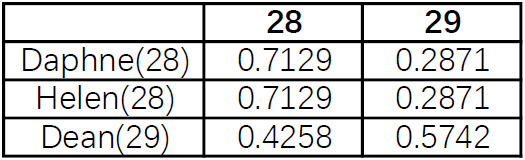}
			\label{fig_3a}
	\end{minipage}}
	\subfigure[]{
		\begin{minipage}[t]{0.45\linewidth}
			\centering
			\includegraphics[width=1.3in]{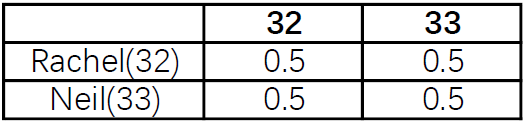}
			\label{fig_3b}
	\end{minipage}}
	\subfigure[]{
		\begin{minipage}[t]{0.45\linewidth}
			\centering
			\includegraphics[width=1.3in]{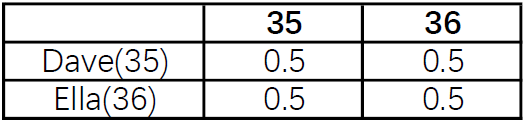}
			\label{fig_3c}
	\end{minipage}}
	\subfigure[]{
		\begin{minipage}[t]{0.45\linewidth}
			\centering
			\includegraphics[width=1.3in]{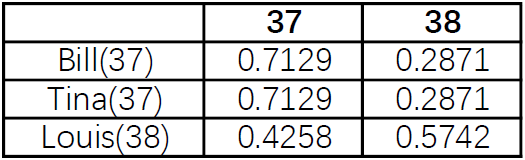}
			\label{fig_3d}
	\end{minipage}}
	\caption{The random output tables.}
	\label{fig_3}
\end{figure}

Additionally, differing from traditional principles that directly confine the values in microdata, we propose a $\delta$-probability principle to control random output tables so as to limit the probability of any QI value being used to re-identify a target person. For instance, the random output tables in Figure \ref{fig_3} comply with $\frac{1}{2}$-probability, and every record is re-identified by using her or his age value with a low probability.

The advantages of MuCo are summarized as follows. First, MuCo can maintain the distributions of original QI values as much as possible. For instance, the sum of each column in Figure \ref{fig_3} is shown by the blue polyline in Figure \ref{fig_2}, and the blue polyline almost coincides with the red polyline representing the distribution in the original data. Second, the anonymization of MuCo is a ``black box'' process for recipients because the only difference between the original data and the anonymized data is that some original QI values are replaced with random values. Thus, the adversary cannot determine which QI values are altered as well as the ranges of variations, causing that the matching tuples are more likely to be wrong or even does not exist when the adversary uses more QI values to match, but the adversary obtains much more matching records if the size of the combination of QI values is not big enough. While for the recipient, the results of query statements are specific records rather than groups. Accordingly, the results are more accurate. The conducted extensive experiments also illustrate the effectiveness of the proposed method.

\section{Related Work}
\label{sec_rela}
\subsection{$k$-Anonymity Based Methods}
The major research of privacy preservation focuses on preventing various disclosures and studying the trade-off between privacy protection and information preservation \cite{li2009on, wong2011can, LIN20099784, li2017cross, KACHA20224075}. The generalization technique has been well-studied by proposing numerous algorithms which can be divided into three schemes: (1) global recoding \cite{lefevre2005incognito}, which transforms same QI values into same generalized value; (2) local recoding \cite{xu2006utility}, which can transform same QI values into different generalized values; and (3) multi-dimensional recoding \cite{lefevre2006mondrian}, which partitions microdata into equivalence groups and generalizes the QI values inside each equivalence group. However, generalization hardly preserves the distributions of original QI values that always causes a huge cost of protecting privacy. 

Comparing to generalization, bucketization technique \cite{xiao2006anatomy, li2021local} maintains excellent information utility because it preserves all the original QI values. However, most existing approaches cannot prevent identity disclosure, and the existence of individuals in published table is likely to be disclosed \cite{nergiz2007hiding}. Furthermore, the QI values of individuals can be easily exposed that increases the background knowledge of adversary to learn the pattern of QI values and sensitive values in the released table \cite{kifer2009attacks, sei2017anonymization}.

\subsection{Differential Privacy}
Differential privacy \cite{dwork2006differential, zhang2014towards}, which is proposed for query-response systems, prevents the adversary from inferring the presence or absence of any individual in the database by adding random noise (e.g., Laplace Mechanism \cite{dwork2006calibrating} and Exponential Mechanism \cite{mcsherry2007mechanism}) to aggregated results. However, differential privacy also faces the contradiction between privacy protection and data analysis \cite{han2015sensitive}. For instance, a smaller $\epsilon$ for $\epsilon$-differential privacy provides better protection but worse information utility.

In recent years, local differential privacy \cite{kasiviswanathan2011what, bun2019heavy} has attracted increasing attention because it is particularly useful in distributed environments where users submit their sensitive information to untrusted curator. Randomized response \cite{holohan2017optimal} is widely applied in local differential privacy to collect users' statistics without violating the privacy. Inspired by local differential privacy, this paper uses the method of randomized response to perturb original QI values before release to prevent the disclosure of matching the combination of QI values. 

Note that, the application scenarios of differential privacy and the models of $k$-anonymity family are different. Differential privacy adds random noise to the answers of the queries issued by recipients rather than publishing microdata. While the approaches of $k$-anonymity family sanitize the original microdata and publish the anonymized version of microdata. Therefore, differential privacy is inapplicable to the scenario we addressed in this paper.

\section{The MuCo Model}
\label{sec_proba}
Suppose that a microdata table $T$ consists of $d$ QI attributes, denoted as $A^{QI}_1, A^{QI}_2, \cdots, A^{QI}_d$, and a sensitive attribute, denoted as $A^{SA}$. Each attribute can be either categorical or continuous, and $D[A]$ represents the domain of attribute $A$. For any tuple $t \in T$, $t[A]$ represents the value of $t$ on attribute $A$.

\subsection{Formalization}
\label{sec_for}
Given a set of tuples, MuCo partitions the tuples into groups, calculates a random output table on each QI attribute inside each group, and generates random values to replace the original QI values according to the random output tables. The formalization is as follows.

\begin{definition}[Random Output Table]
	\label{def_table}
	A random output table contains the probabilities that the records output the corresponding QI values. Suppose that a group includes $m$ tuples $\{t_{1}, t_{2}, \cdots, t_{m}\}$, and the output values are represented as $\{v_{1}, v_{2}, \cdots, v_{n}\}$. The random output table forms a $m \times n$ matrix  
	\begin{displaymath}
		\label{equ_matrix}
		\left(\begin{array}{cccc}
			p_{11} & p_{12} & \cdots & p_{1n}\\
			\vdots & \ddots &        & \vdots\\
			p_{m1} & p_{m2} & \cdots & p_{mn}
		\end{array}\right),
	\end{displaymath}
	where $p_{ij}$ represents the probability that $t_{i}$ carries $v_{j}$, such that
	\begin{displaymath}
		p_{ij}=p(v^{i}_{output}=v_{j}|t_{i}), i \in [1,m]\ and\ j \in [1, n],
	\end{displaymath}
	where $v^{i}_{output}$ denotes the random output value for $t_{i}$. Moreover, we also have
	\begin{displaymath}
		p_{ij} \geq 0, i \in [1, m]\ and\ j \in [1, n],
	\end{displaymath}
	and
	\begin{displaymath}
		\sum^{n}_{j=1}p_{ij}=1, ~~i \in [1, m],
	\end{displaymath}
	such that each probability must be non-negative, and every record can only carry a random value within $\{v_{1}, v_{2}, \cdots, v_{n}\}$.
\end{definition}

\begin{definition}[Mutual Cover Strategy]
	\label{def_mutual}
	Given a microdata table $T$, the mutual cover strategy partitions $T$ into groups, calculates a random output table on each QI attribute for the records in every group, and generates random values according to probabilities in the random output tables.
\end{definition}

\begin{figure}[t]
	\centering
	\subfigure[]{
		\begin{minipage}[t]{0.45\linewidth}
			\centering
			\includegraphics[height=0.8in]{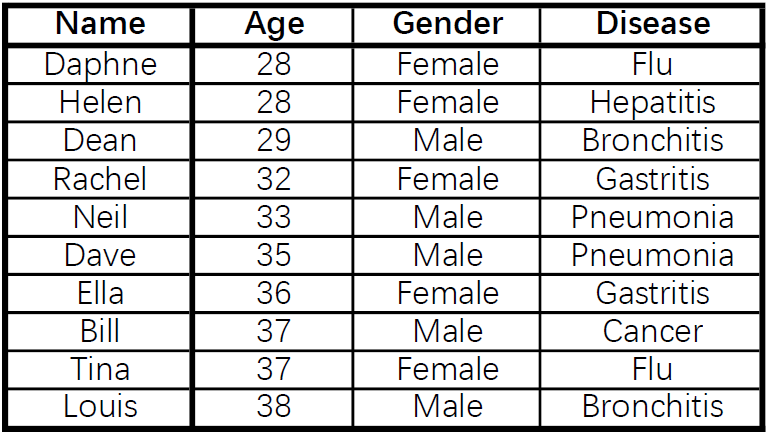}
			\label{fig_4a}
	\end{minipage}}
	\subfigure[]{
		\begin{minipage}[t]{0.45\linewidth}
			\centering
			\includegraphics[height=0.8in]{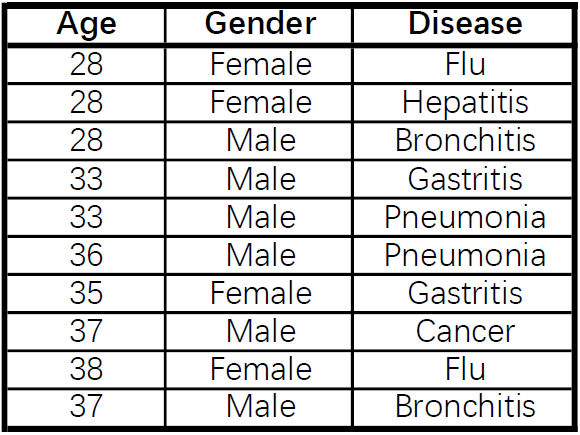}
			\label{fig_4b}
	\end{minipage}}
	\caption{An example of MuCo.}
	\label{fig_4}
\end{figure}

For instance, suppose that we add another QI attribute of gender as shown in Figure \ref{fig_4a}, the mutual cover strategy first divides the records into groups in which the records in the same group cover for each other by perturbing their QI values. Then, the mutual cover strategy calculates a random output table on each QI attribute (i.e., age and gender) within each group. Finally, an anonymized table, as shown in Figure \ref{fig_4b}, is generated by replacing the original QI values with the random output values. Note that, since the anonymization process is hidden, the adversary does not know the partition of groups and the random output tables. Therefore, the adversary can not determine which QI values are changed as well as the ranges of the variations.

Next, a principle called $\delta$-probability is proposed to control the probabilities in the random output tables.

\begin{definition}[$\delta$-Probability]
	\label{def_skpro}
	For an anonymized table of MuCo, it complies with $\delta$-probability, if all the random output tables satisfy the following equation
	\begin{displaymath}
		\frac{max(p_{j})}{\sum^{m}_{i=1}p_{ij}} \leq \delta, ~~j \in [1, n],
	\end{displaymath}
	where $max(p_{j})$ is the maximum probability in column $j$, $m$ is the number of records in group, and $n$ is the size of output values.
\end{definition}

$\delta$-Probability limits the probability, to a large extent, that an adversary re-identifies a target person by matching a QI value.

\begin{example}
	Suppose that an adversary aims to find the record of Helen in the anonymized table by matching her age value of 28, and the anonymization process is hidden for the adversary. Then, the adversary does not know the partition of groups and the random output table including the range and probabilities of the random output values of Helen. According to Figure \ref{fig_3}, there are three records, namely Daphne, Helen, and Dean, may carry 28. Therefore, the probability that Helen is re-identified by matching her age value of 28 is calculated as
	
	\begin{footnotesize}
		\begin{displaymath}
			p(Helen) = \sum_{i=1}^{3}\frac{p(v_{Helen}=28,num(28)=i)}{i} \approx 0.3741,
		\end{displaymath}
	\end{footnotesize}
	where $num(28)$ denotes the number of 28s in the anonymized table, and $v_{Helen}$ represents the output value of Helen. 
	
	As a result, Helen maintains her original QI value with a high probability (i.e., 0.7129) but being re-identified by matching her age value with a low probability.
\end{example}

Note that, in the worst case, although the curator does not publish the details of anonymization, the adversary may know that the anonymized table is sanitized by MuCo. Therefore, the adversary may find the record of target person by using approximate QI values. However, the adversary must obtains excessive matching records which definitely decreases the accuracy of re-identifying the target person, especially in high-dimensional microdata.

\subsection{Analysis on Parameter $\delta$}
\label{sec_ana}
In this subsection, we discuss the property and effect of $\delta$ on the anonymized table complying with $\delta$-probability.

\begin{property}
	\label{proper_range}
	Given an anonymized table of MuCo complying with $\delta$-probability, for any group, $\delta$ must be in the range of $[1/m, 1]$, where $m$ is the number of tuples in the group.
\end{property}

\begin{proof}
	According to Definition \ref{def_skpro}, for any column $j$ in the random output table, we have 
	\begin{displaymath}
		\frac{max(p_{j})}{\sum^{m}_{i=1}p_{ij}} \leq \delta,
	\end{displaymath}
	where $m$ is the number of tuples in the group. For any column $j$, we have
	\begin{displaymath}
		max(p_{j}) \leq \sum^{m}_{i=1}p_{ij} \leq m \cdot max(p_{j}).
	\end{displaymath}
	Thus,
	\begin{displaymath}
		\frac{max(p_{j})}{m \cdot max(p_{j})} \leq \frac{max(p_{j})}{\sum^{m}_{i=1}p_{ij}} \leq \frac{max(p_{j})}{max(p_{j})},
	\end{displaymath}
	then we have
	\begin{displaymath}
		\frac{1}{m} \leq \frac{max(p_{j})}{\sum^{m}_{i=1}p_{ij}} \leq 1.
	\end{displaymath}
	Therefore, we conclude that
	\begin{displaymath}
		\delta \geq \frac{max(p_{j})}{\sum^{m}_{i=1}p_{ij}} \geq \frac{1}{m},
	\end{displaymath}
	and $\delta$ does not affect random output tables when $\delta$ is greater than 1, such that $\delta$ must be in the range of $[1/m, 1]$. \qed
\end{proof}

Property \ref{proper_range} demonstrates the constraint that the range of $\delta$ depends on the number of tuples in the group. Next, the relation between the value of $\delta$ and the number of correlative tuples, given a released QI value, is discussed as follows.

\begin{corollary}
	\label{cor_kn}
	Given an anonymized table of MuCo complying with $\delta$-probability, each QI value in the released table corresponds to at least $\lceil 1/\delta \rceil$ records.
\end{corollary}

\begin{proof}
	According to Definition \ref{def_skpro}, for any column $j$ in the random output table, we have 
	\begin{displaymath}
		\frac{max(p_{j})}{\sum^{m}_{i=1}p_{ij}} \leq \delta,
	\end{displaymath}
	where $m$ is the number of tuples in the group. Thus, we have
	\begin{displaymath}
		\sum^{m}_{i=1}p_{ij} \geq \frac{1}{\delta} \cdot max(p_{j}).
	\end{displaymath}
	For any QI value in the released table, the corresponding column $j$ in the random output table must have
	\begin{displaymath}
		max(p_{j}) > 0.
	\end{displaymath}
	According to Property \ref{proper_range}, $\delta$ is within $[1/m, 1]$. Let $k$ denote $\lfloor 1/\delta \rfloor$, then 
	\begin{displaymath}
		\begin{cases}
			k \cdot max(p_{j}) < \frac{1}{\delta} \cdot max(p_{j}) & \text{if $\lfloor 1/\delta \rfloor \neq \lceil 1/\delta \rceil$} \\
			k \cdot max(p_{j}) = \frac{1}{\delta} \cdot max(p_{j}) & \text{if $\lfloor 1/\delta \rfloor = \lceil 1/\delta \rceil$}
		\end{cases}.
	\end{displaymath}
	Therefore, we have
	\begin{displaymath}
		\begin{cases}
			\sum^{m}_{i=1}p_{ij} > k \cdot max(p_{j}) & \text{if $\lfloor 1/\delta \rfloor \neq \lceil 1/\delta \rceil$} \\
			\sum^{m}_{i=1}p_{ij} \geq k \cdot max(p_{j}) & \text{if $\lfloor 1/\delta \rfloor = \lceil 1/\delta \rceil$}
		\end{cases}.
	\end{displaymath}
	Thus, there are at least $\lceil 1/\delta \rceil$ probabilities greater than 0 in the column, such that each QI value in the released table corresponds to at least $\lceil 1/\delta \rceil$ records. \qed
\end{proof}

For instance, since the random output tables in Figure \ref{fig_3} comply with $\frac{1}{2}$-probability, for any QI value whose corresponding column has at least one probability greater than 0, there are at least 2 records can carry the QI value.

\section{The MuCo Algorithm}
\label{sec_algo}
This section presents the algorithm to implement the Mutual Cover (MuCo) framework\footnote{The code is available at https://github.com/liboyuty/Mutual-Cover.}. We aim to achieve two goals. First, MuCo satisfies $\delta$-probability to hinder the adversary from matching the combination of QI values. Second, the records cover for each other at the minimum cost, i.e., maintaining the original QI values as much as possible. The procedure is given in Algorithm \ref{alg_mc}.

\begin{algorithm}
	\caption{MuCo$(T,\delta)$}
	\label{alg_mc}
	\begin{algorithmic}[1]
		\STATE $T_{anony}=T$
		\STATE $groups=divide\_table(T_{anony})$
		\FOR {\bf{each} $group \in groups$}
		\FOR {\bf{each} $attri \in \{A^{QI}_1, A^{QI}_2, \cdots, A^{QI}_d\}$}
		\STATE $rtable=calculate\_table(group, attri, \delta)$
		\STATE $output\_values(group, attri, rtable)$
		\ENDFOR
		\STATE $randomize\_unchanged\_tuples(group)$
		\ENDFOR
		\STATE \bf{return} $T_{anony}$
	\end{algorithmic}
\end{algorithm}

$T_{anony}$ denotes the anonymized version of $T$, and we use Mondrian \cite{lefevre2006mondrian} to divide $T_{anony}$ into groups (line 2). Note that, to prevent the attribute disclosure, the partition has to follow an appropriate principle in case the similar records in the same group do not carry enough different sensitive values. In each iteration, the algorithm calculates the random output table on each QI attribute within each group (line 5) and replaces the original QI values with random values in $T_{anony}$ (line 6). $randomize\_unchanged\_tuples(group)$ randomly perturbs a QI value for the tuples whose all QI values are originally preserved. Note that, MuCo only replaces some original QI values with random values in $T_{anony}$. Therefore, the whole procedure is hidden for the adversary. 

Next, we elaborate $calculate\_table(group, attri, \delta)$ and $randomize\_unchanged\_tuples(group)$, respectively. To calculate the random output tables that can retain original QI values as much as possible, we first define the distance function as follows.

\begin{definition}[Distance Function]
	\label{def_distance}
	The distance function, denoted as $dis(\cdot)$, measures the distance between two values, such that for any two values $v_{1}$ and $v_{2}$, $dis(v_{1},v_{2}) \in \mathbb{R}^{+}_{0}$, where $\mathbb{R}^{+}_{0}$ represents the region of non-negative real number.
\end{definition}

Note that, it is difficult to define the distance between QI values on categorical attribute. One feasible solution is to build a customized hierarchy tree for each categorical attribute and assign a specific distance value for each pair of leaf nodes. While for continuous attribute, we implement $dis(\cdot)$ by $l_{1}$-distance defined as follows.

\begin{definition}[$l_{1}$-Distance]
	\label{def_l1}
	Given two values $v_{1}$ and $v_{2}$, their $l_{1}$-distance is expressed as
	\begin{displaymath}
		||v_{1}-v_{2}||_{1}=|v_{1}-v_{2}|.
	\end{displaymath}
\end{definition}

We transform the problem to a linear programming model and use the primal-dual path following algorithm \cite{andersen2000the} to calculate the random output tables. According to Definition \ref{def_table} and Definition \ref{def_skpro}, the linear programming for implementing $calculate\_table(group, attri, \delta)$ is given as
\begin{align*}
	\text{Min} &~~~~\sum_{i=1}^{m}\sum_{j=1}^{n}dis(v^{ori}_{i},v^{col}_{j})p_{ij} \label{equ_1}\\
	\textit{s.t.~~} &(1)~~ \sum_{j=1}^{n}p_{ij}=1, ~~i \in [1, m], \\
	& (2)~~\frac{max(p_{j})}{\sum_{i=1}^{m}p_{ij}} \leq \delta, ~~j \in [1, n], \\
	& (3)~~p_{ij} \geq 0, ~~i \in [1, m]\ and\ j \in [1, n].
\end{align*}
Here $v^{ori}_{i}$ is the original value of $t_{i}$ on attribute $attri$, $v^{col}_{j}$ is the value of column $j$, $m$ is the number of records in $group$, $n$ is the size of domain within $group$ on attribute $attri$. The meaning of the objective function is to assign the probability of each record to the output values (i.e., the values within the domain of $attri$ inside $group$), and make the values, which have  longer distances to the original values of the record, carry lower probabilities on the premise of satisfying the constraints.

Finally, $randomize\_unchanged\_tuples(group)$ is described in detail. There are some tuples that may output all the original QI values in the anonymized table due to randomness, causing that their identities may be disclosed with high probability. Therefore, we perturb a QI value for each unchanged tuple\footnote{The number of perturbed QI values can be set according to actual demands.}. Algorithm \ref{alg_rut} provides the process of $randomize\_unchanged\_tuples(group)$.

\begin{algorithm}
	\caption{randomize\_unchanged\_tuples$(group)$}
	\label{alg_rut}
	\begin{algorithmic}[1]
		\STATE $attribute\_weights=count\_weights(group)$
		\STATE $unchanged\_tuples=pick\_unchanged(group)$
		\FOR {\bf{each} $tuple \in unchanged\_tuples$}
		\WHILE {$check\_unchanged(tuple)$}
		\STATE $attri=choose\_attri(attribute\_weights)$
		\STATE $replace\_QI\_value(tuple, attri, group)$
		\ENDWHILE
		\ENDFOR
		\STATE \bf{return} $group$
	\end{algorithmic}
\end{algorithm}

$attribute\_weights$ denotes the weights of QI attributes, and each weight is calculated by $\frac{dis_{max}(attri,group)}{dis_{max}(attri,T)}$, where $dis_{max}(attri,group)$ and $dis_{max}(attri,T)$ represent the max distance between values on attribute $attri$ within $group$ and $T$, respectively. $unchanged\_tuples$ is the set of tuples whose all QI values are originally preserved. In each iteration, the algorithm chooses a tuple from $unchanged\_tuples$, picks an attribute according to $attribute\_weights$, and replaces an original QI value with a random value from the range within $group$ on attribute $attri$ until at least one QI value is different from the original value.

\section{Experimental Results}
\label{sec_expe}
This section evaluates the effectiveness of the proposed MuCo algorithm. We apply Mondrian \cite{lefevre2006mondrian}, which is one of the most effective generalization approaches, and Anatomy \cite{xiao2006anatomy}, which always preserves the best information utility, as the  baselines. We use the US Census data \cite{data}, eliminate the tuples with missing values, and randomly select 40,152 tuples with eight attributes. The QI attributes are gender, age, relationship, marital status, race, education, and hours per week, and the sensitive attribute is salary. Table \ref{tab_attri} describes the attributes in detail.

\begin{table}[!t]
	\centering
	\caption{Description of the attributes.}
	\label{tab_attri}
	\begin{tabular}{|c|c|c|c|}
		\hline
		&\bf{Attribute}&\bf{Type}&\bf{Size}\\
		\hline
		1 & Gender & Categorical & 2\\
		\hline
		2 & Age & Continuous & 55\\
		\hline
		3 & Relationship & Categorical & 13\\
		\hline
		4 & Marital status & Categorical & 6\\
		\hline
		5 & Race & Categorical & 9\\
		\hline
		6 & Education & Categorical & 10\\
		\hline
		7 & Hours per week & Continuous & 95\\
		\hline
		8 & Salary & Continuous & 851\\
		\hline
	\end{tabular}
\end{table}

\subsection{Privacy Protection}
\label{sec_privacy}
This experiment measures the effectiveness of privacy preservation for MuCo. We assume that an adversary has all the QI values of the target person, and each QI value is combined to match the target person with the probability of $P_{match}$. This experiment counts the number of matching records for all the tuples in the anonymized tables, and the average probabilities of identity disclosure and attribute disclosure per tuple are calculated. Let MuCo comply with $\delta$-probability and the partition of microdata (i.e., $divide\_table(T_{anony})$ in Algorithm \ref{alg_mc}) satisfy $l$-diversity, where the parameter $\delta$ is set at $\frac{1}{5}$, $\frac{1}{6}$, $\frac{1}{7}$, $\frac{1}{8}$ and $\frac{1}{10}$, respectively, and $l$ is assigned to 10. To reduce the accidental effect of randomness, we run MuCo with the same parameters for 10 times, the results are shown in Figure \ref{fig_dis}. Additionally, since MuCo generates 10 anonymized tables in every experiment, we use boxplot to report the average value and variance at the same time for MuCo.

\begin{figure*}[!t]
	\centering
	\subfigure[$P_{match}=0.3$]{\includegraphics[width=1.75in]{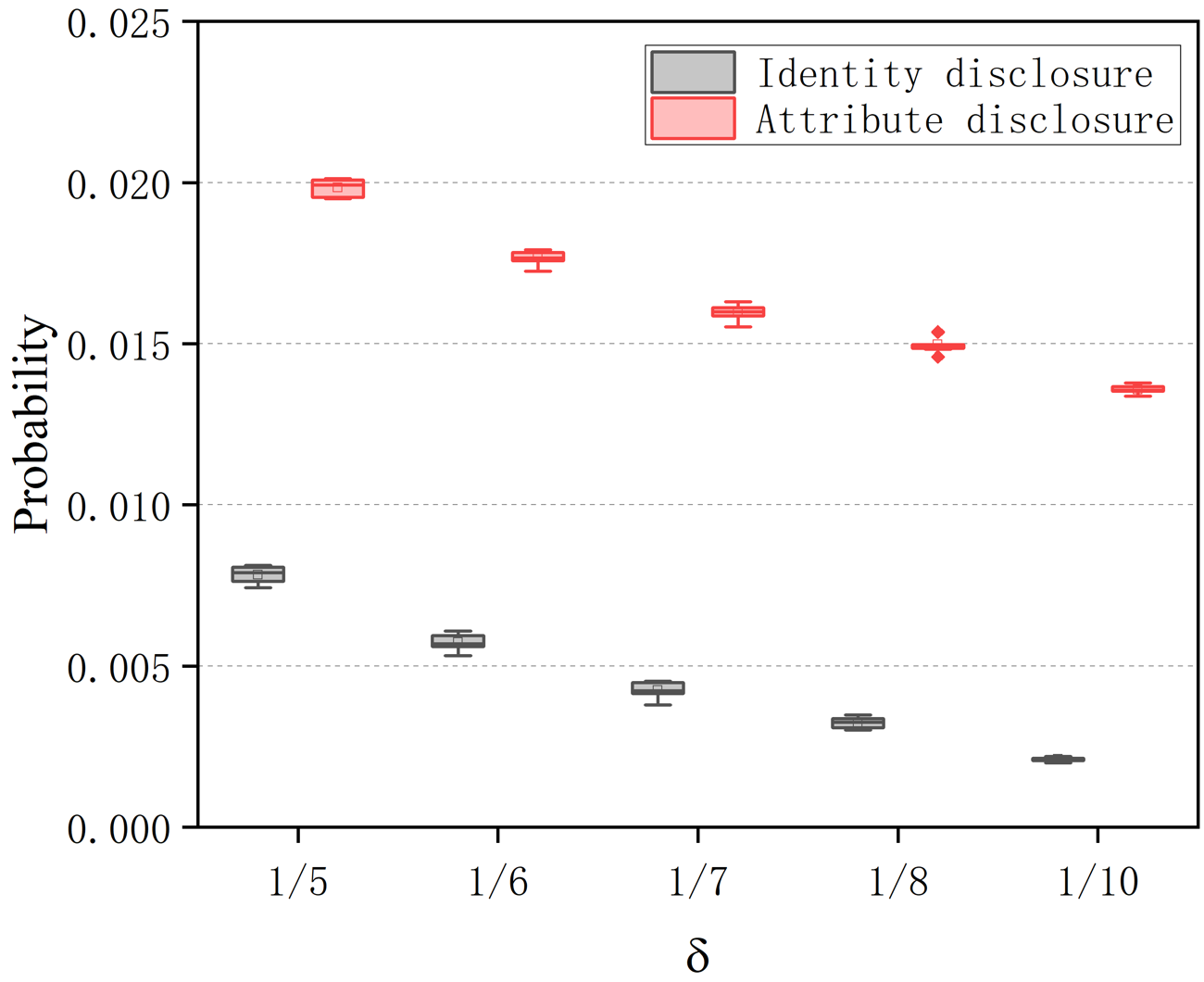}
		\label{fig_iden_dis_0.3}}
	\subfigure[$P_{match}=0.5$]{\includegraphics[width=1.75in]{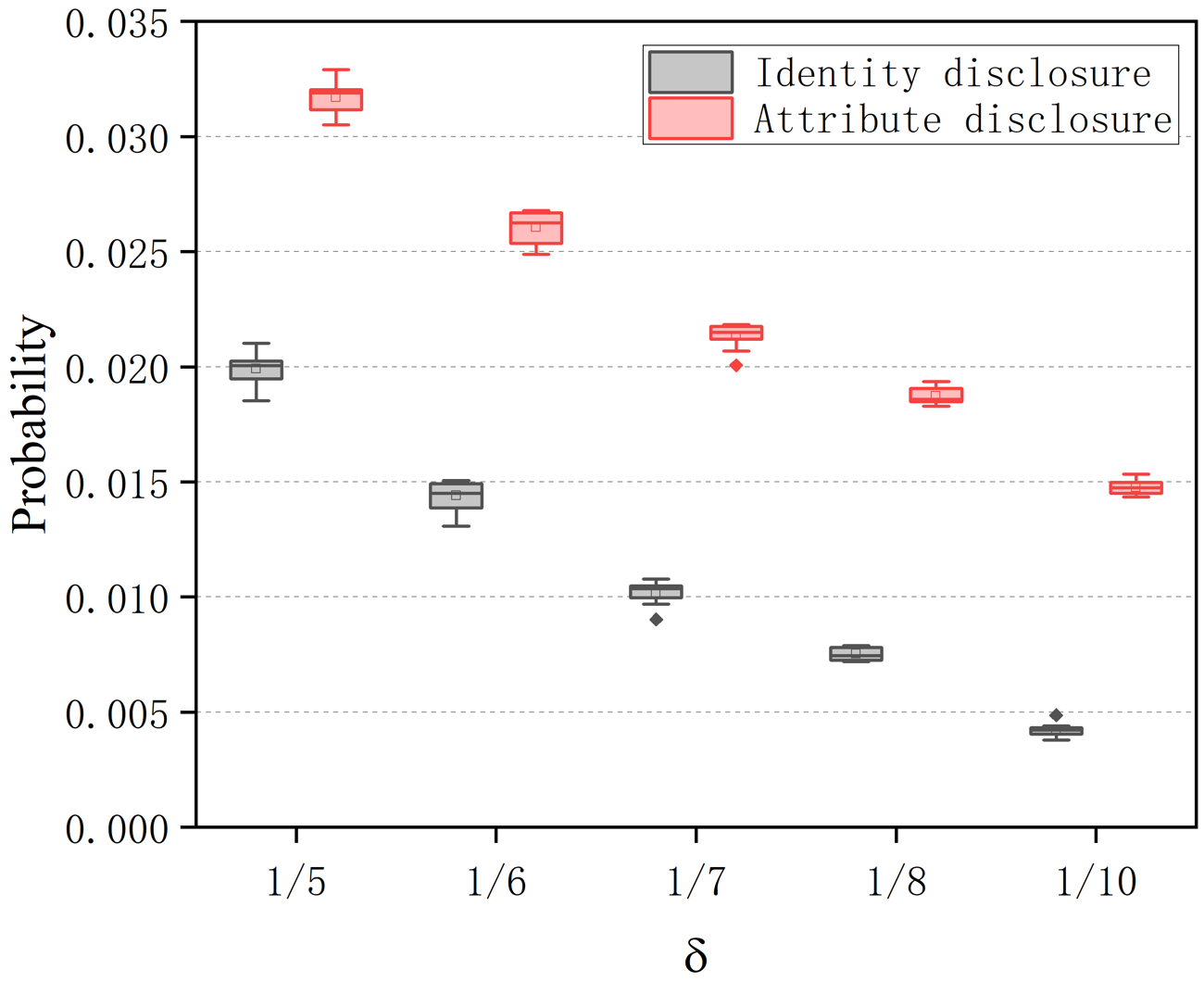}
		\label{fig_iden_dis_0.5}}
	\subfigure[$P_{match}=0.7$]{\includegraphics[width=1.75in]{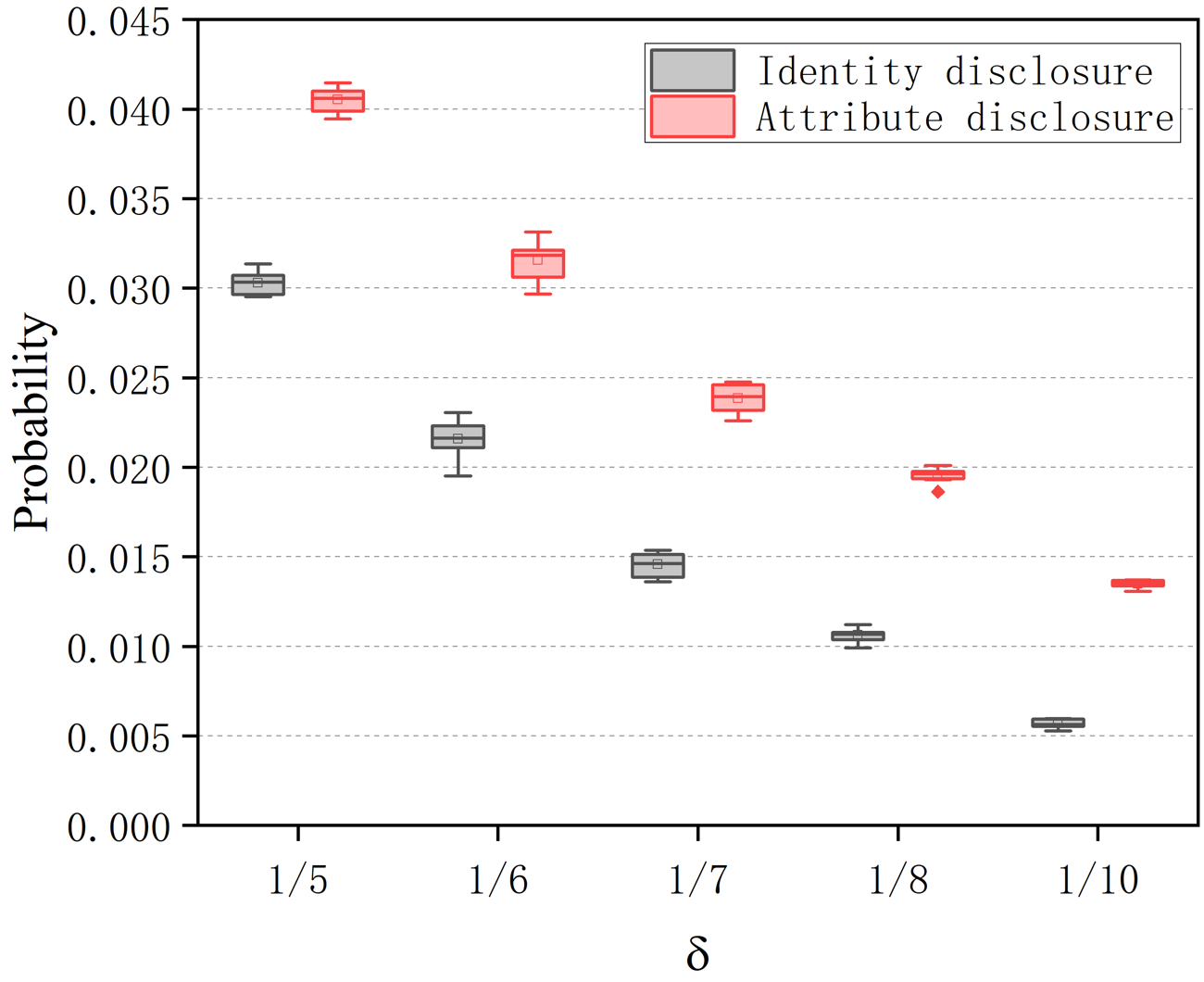}
		\label{fig_iden_dis_0.7}}
	\subfigure[$P_{match}=0.8$]{\includegraphics[width=1.75in]{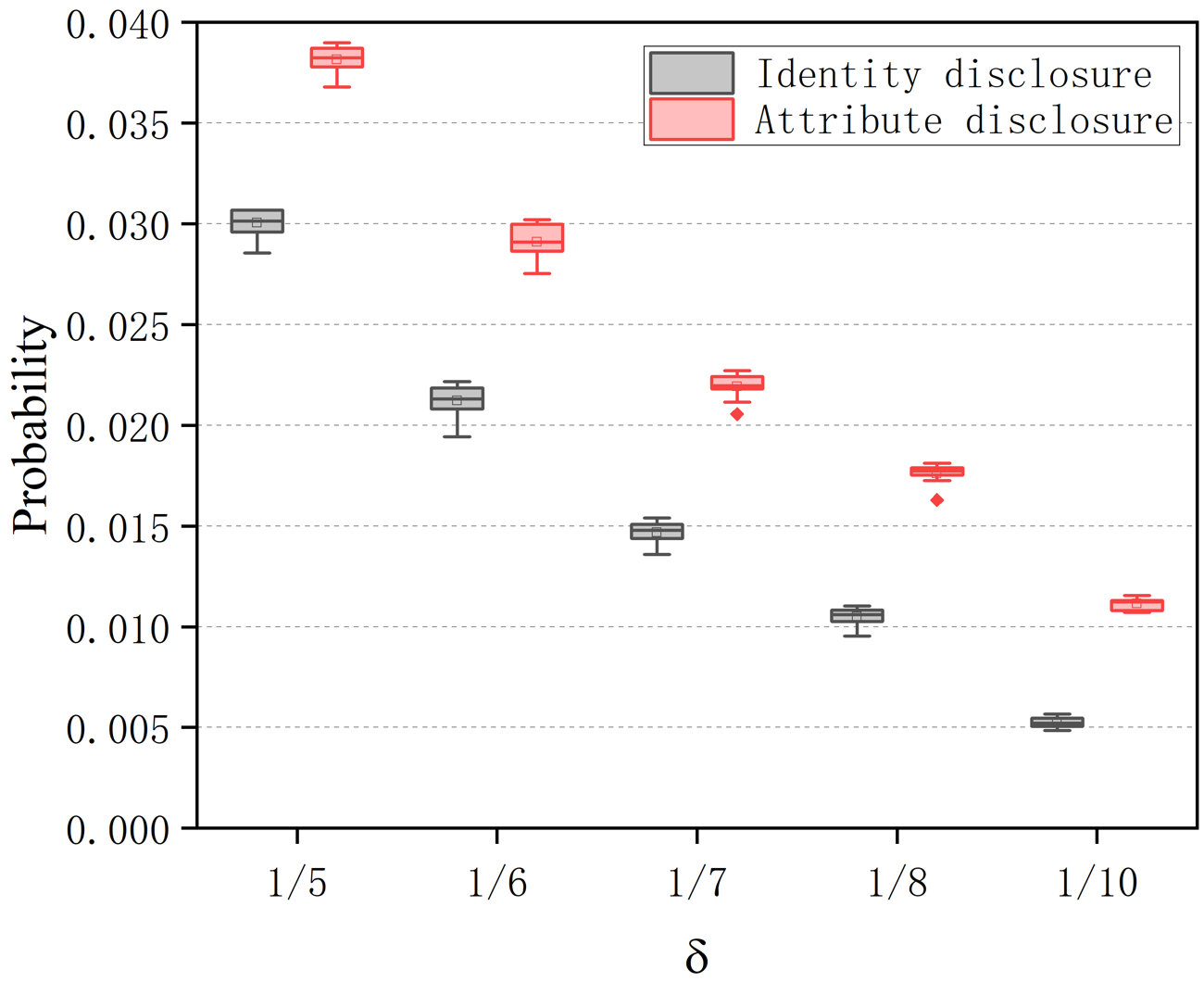}
		\label{fig_iden_dis_0.8}}
	\subfigure[$P_{match}=0.9$]{\includegraphics[width=1.75in]{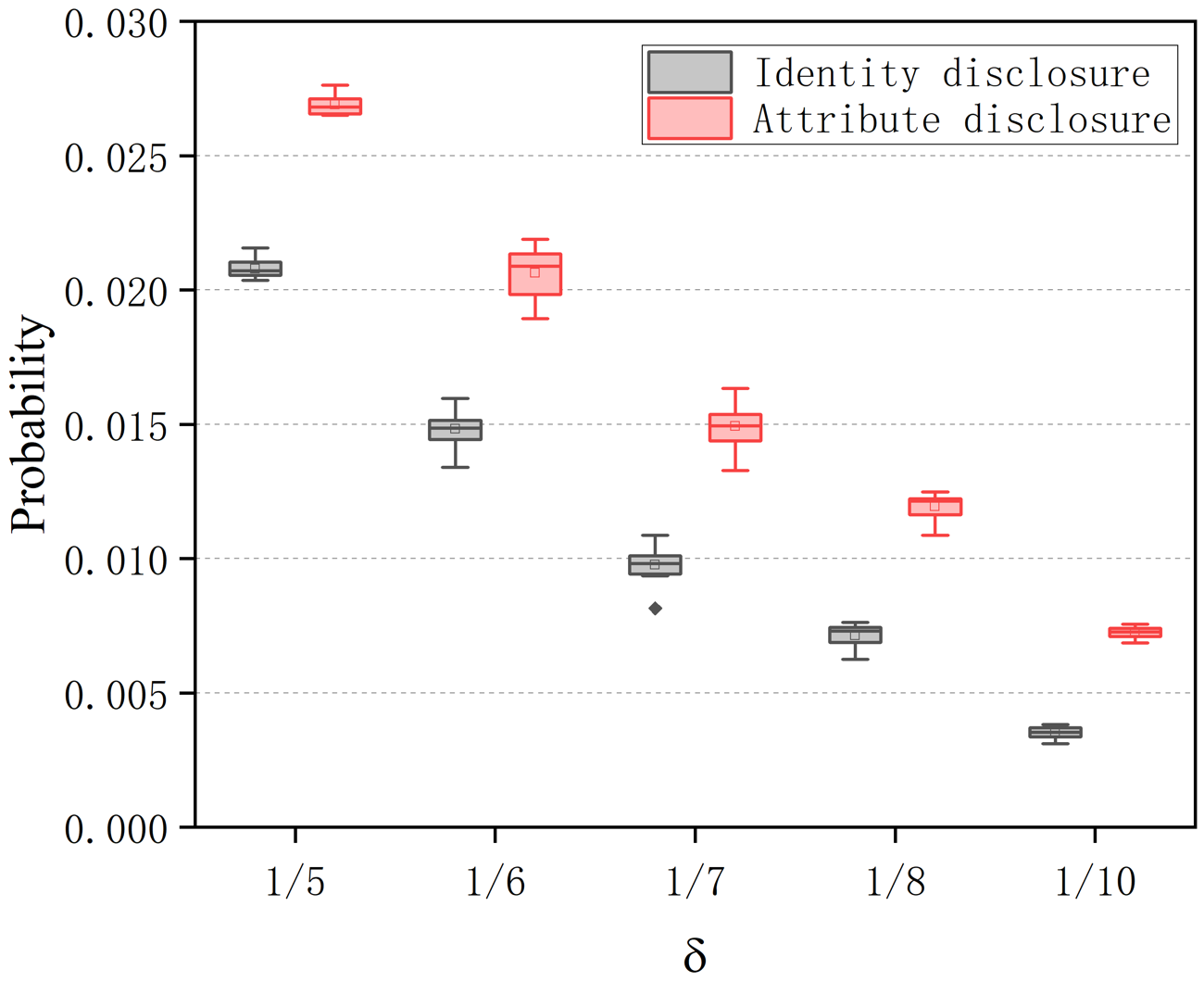}
		\label{fig_iden_dis_0.9}}
	\subfigure[$P_{match}=1$]{\includegraphics[width=1.75in]{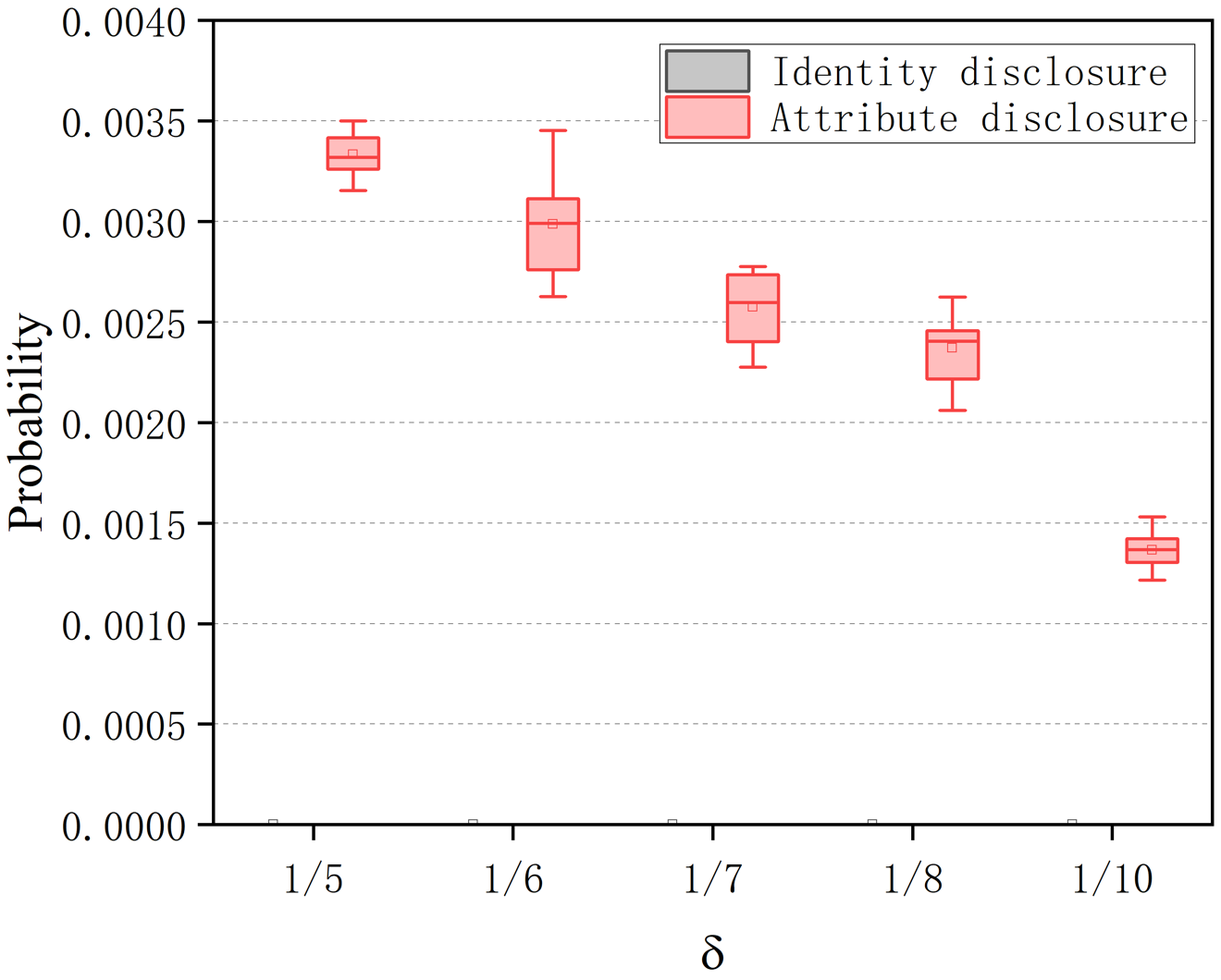}
		\label{fig_iden_dis_1}}
	\caption{Disclosure probability of MuCo.}
	\label{fig_dis}
\end{figure*}

\begin{figure}[!t]
	\centering
	\includegraphics[width=1.75in]{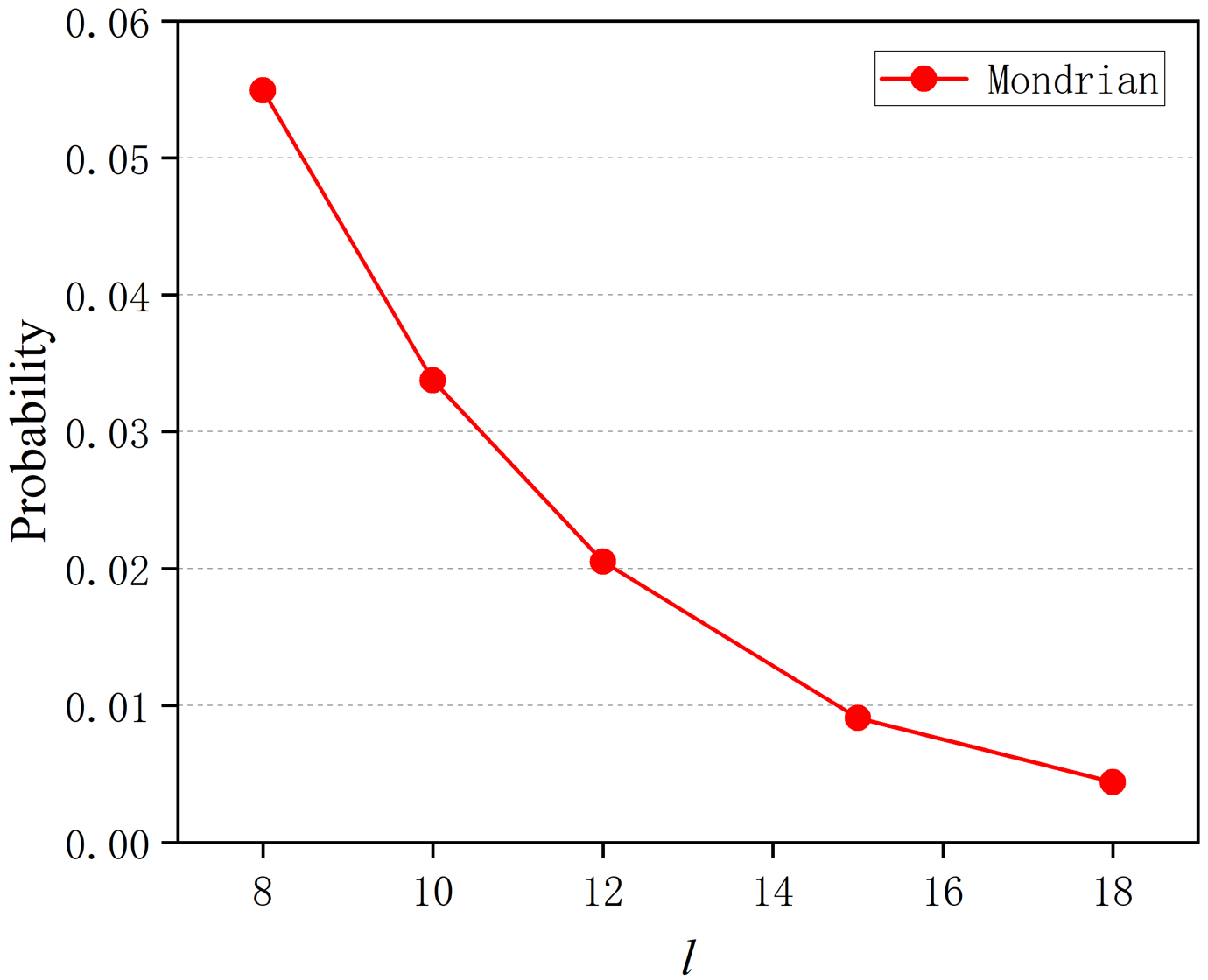}
	\caption{Disclosure probability of Mondrian.}
	\label{fig_dis_mon}
\end{figure}

Observing from Figure \ref{fig_dis}, MuCo provides effective protections against identity disclosure and attribute disclosure. The identity disclosure and attribute disclosure decrease with the reduction of $\delta$ because the released QI values correspond to more records with smaller $\delta$ according to Corollary \ref{cor_kn}. Therefore, each record is covered with substantial other records that dramatically reduces the probability of being re-identified. Note that, since we perform $randomize\_unchanged\_tuples(group)$ in Algorithm \ref{alg_mc}, the probabilities of identity disclosure in Figure \ref{fig_iden_dis_1} keep at 0. Moreover, the probability of disclosure increases when $P_{match}$ raises from 0.3 to 0.7, the probabilities are very similar when $P_{match}$ are 0.7 and 0.8, respectively, and the probability reduces when $P_{match}$ is greater than 0.8. Consequently, the experiments demonstrate that the adversary can still be confused even if he has known much information about the target person. Additionally, the adversary cannot determine the optimal $P_{match}$ in practice because the anonymization process is hidden for the adversary.

\subsection{Information Loss}
\label{sec_il}
This experiment measures the information loss of MuCo. Note that, the mechanism of MuCo is much more different from that of generalization. Thus, for the sake of fairness, we compare the information loss of MuCo and Mondrian when they provide the same level of protections. Then, the experiment measures the effectiveness of protection via the information loss, such that a better anonymization algorithm achieves the same level of protection with a smaller cost. We first implement Mondrian complying with $l$-diversity, where $l$ is set at 8, 10, 12, 15, and 18, respectively, and the disclosure probabilities are shown in Figure \ref{fig_dis_mon}. We can find that the disclosure probabilities\footnote{Note that, the probabilities of identity disclosure and attribute disclosure are the same when Mondrian complies with $l$-diversity.} of Mondrian in Figure \ref{fig_dis_mon} are similar with that of MuCo in Figure \ref{fig_iden_dis_0.7}. Then, we can compare the information loss of Mondrian with that of MuCo using the same parameters as the last experiment. A popular metric $ILoss$ \cite{xiao2006personalized} is applied to estimate the information loss, and the results are shown in Figure \ref{fig_pen}.

\begin{figure}[!t]
	\centering
	\subfigure[MuCo]{
		\begin{minipage}[t]{0.45\linewidth}
			\centering
			\includegraphics[width=1.6in]{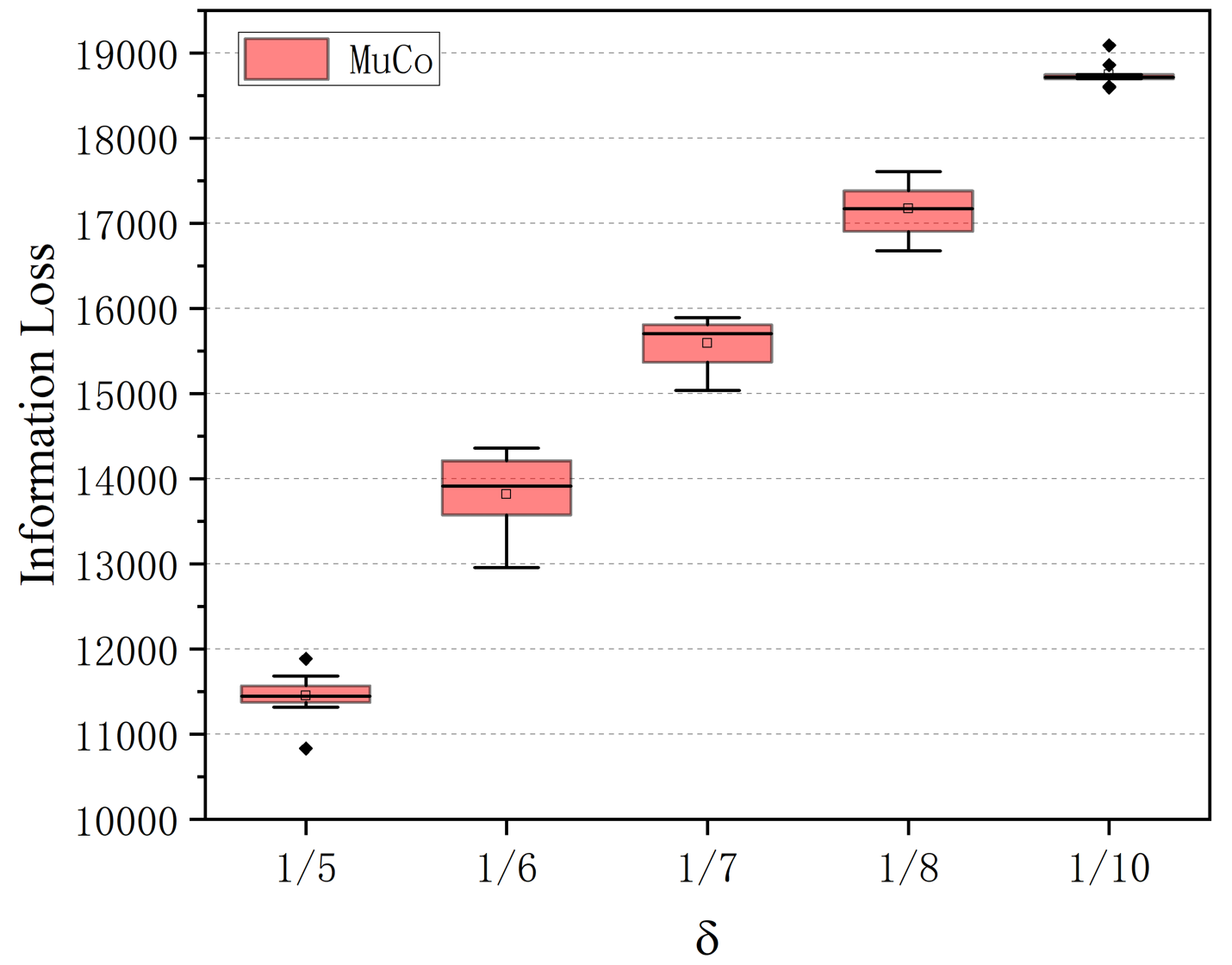}
			\label{fig_pen_mutual}
			\vspace{-0.5em}
	\end{minipage}}
	\hspace{0.1cm}
	\subfigure[Mondrian]{
		\begin{minipage}[t]{0.45\linewidth}
			\centering
			\includegraphics[width=1.6in]{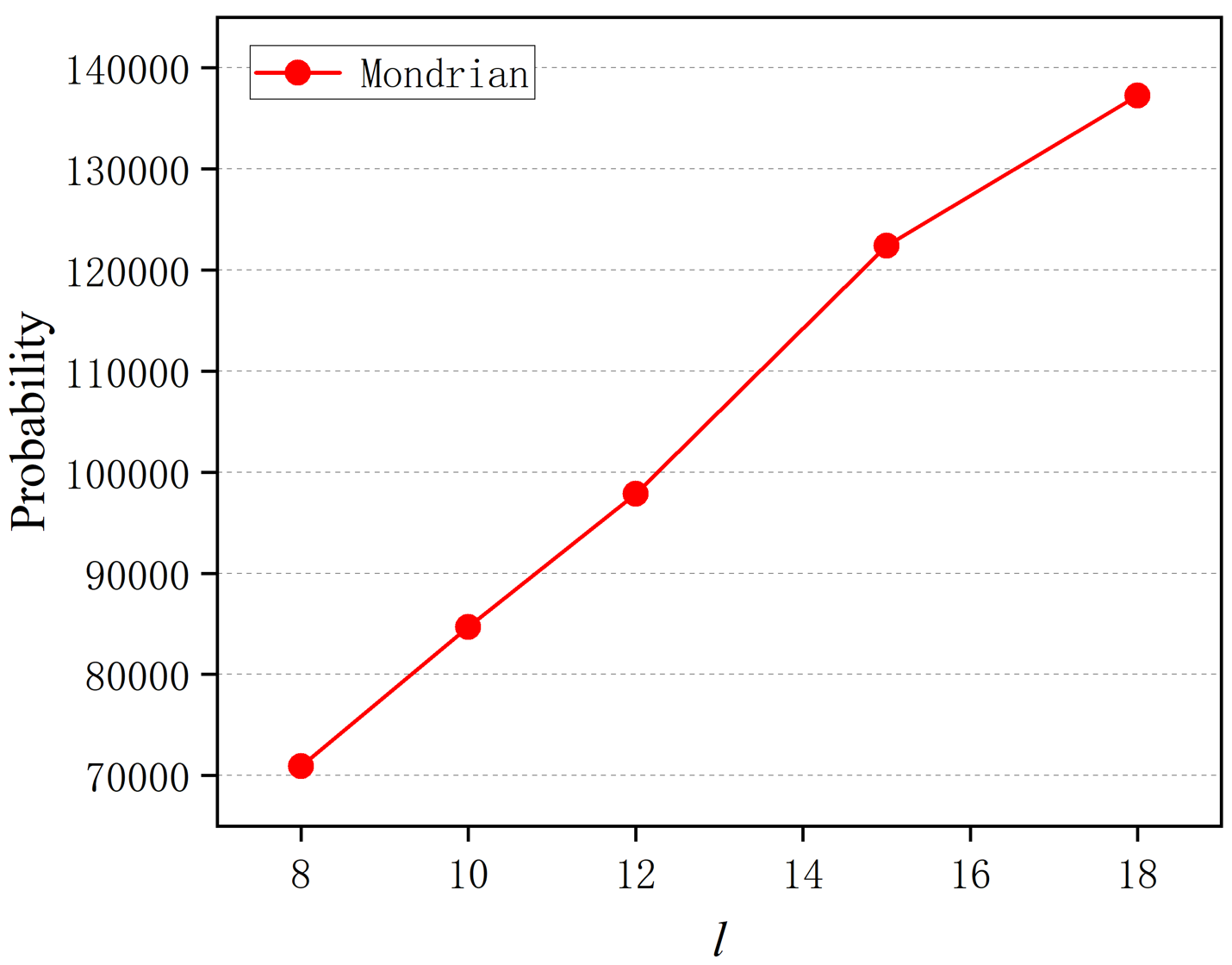}
			\vspace{-0.5em}
			\label{fig_pen_mon}
	\end{minipage}}
	\vspace{-0.5em}
	\caption{Information loss.}
	\label{fig_pen}
\end{figure}

Observing from Figure \ref{fig_pen_mutual}, the information loss of MuCo increases with the decrease of parameter $\delta$. According to Corollary \ref{cor_kn}, each QI value in the released table corresponds to more records with the reduction of $\delta$, causing that more records have to be involved for covering on the QI values of long distance. Therefore, the decrease of $\delta$ enhances the protection but also increases the information loss. In addition, comparing to Figure \ref{fig_pen_mon}, both the information loss and the interval of MuCo are much less than that of Mondrian. Thus, the experiments illustrate that comparing to generalization, MuCo preserves more information utility and enhances the protection at a much smaller cost of information loss.

\subsection{Query Answering}
\label{sec_query}
In this experiment, we use the approach of aggregate query answering \cite{zhang2007aggregate} to check the information utility of MuCo. We randomly generate 1,000 queries and calculate the average relative error rate for comparison. The sequence of the query is expressed in the following form
\begin{flushleft}
	SELECT SUM(salary) FROM \bf{Microdata}\rm\\
	WHERE \bf{$pred$($A_1^{QI}$)} AND \bf{$pred$($A_2^{QI}$)} AND \bf{$pred$($A_3^{QI}$)} AND \bf{$pred$($A_4^{QI}$)}.
\end{flushleft}
Specifically, the query condition contains four random QI attributes, and the sum of salary is the result. We use the same parameters of MuCo and perform Mondrian and Anatomy complying with $l$-diversity for comparison. Since the generated query conditions are strong stochastic, we report the average values and the variances of relative error rates as given in Figure \ref{fig_query} and Figure \ref{fig_var}, respectively.

\begin{figure*}[!t]
	\centering
	\subfigure[MuCo]{\includegraphics[height=1.4in]{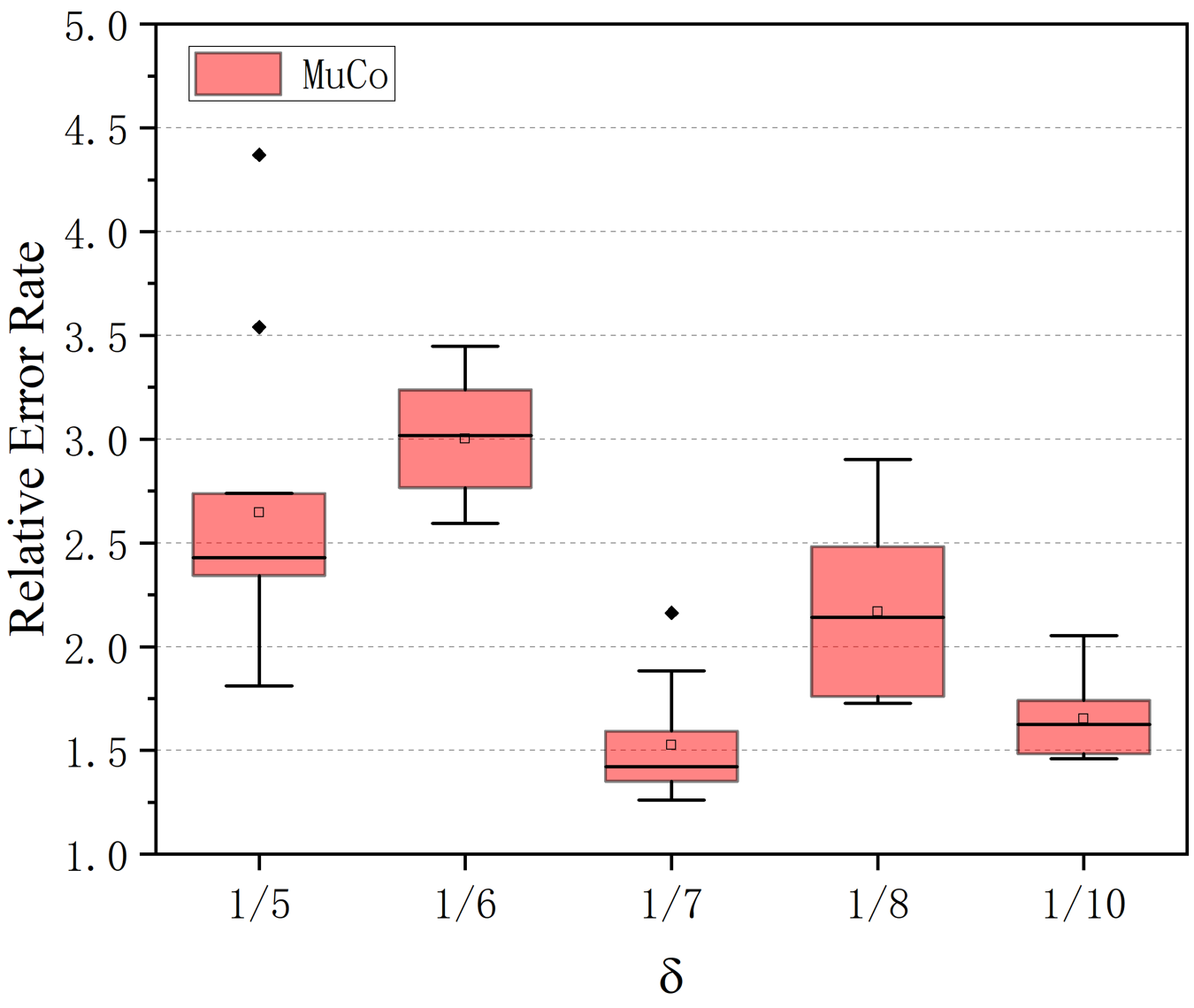}
		\label{fig_query_mutual}}
	\subfigure[Mondrian]{\includegraphics[height=1.4in]{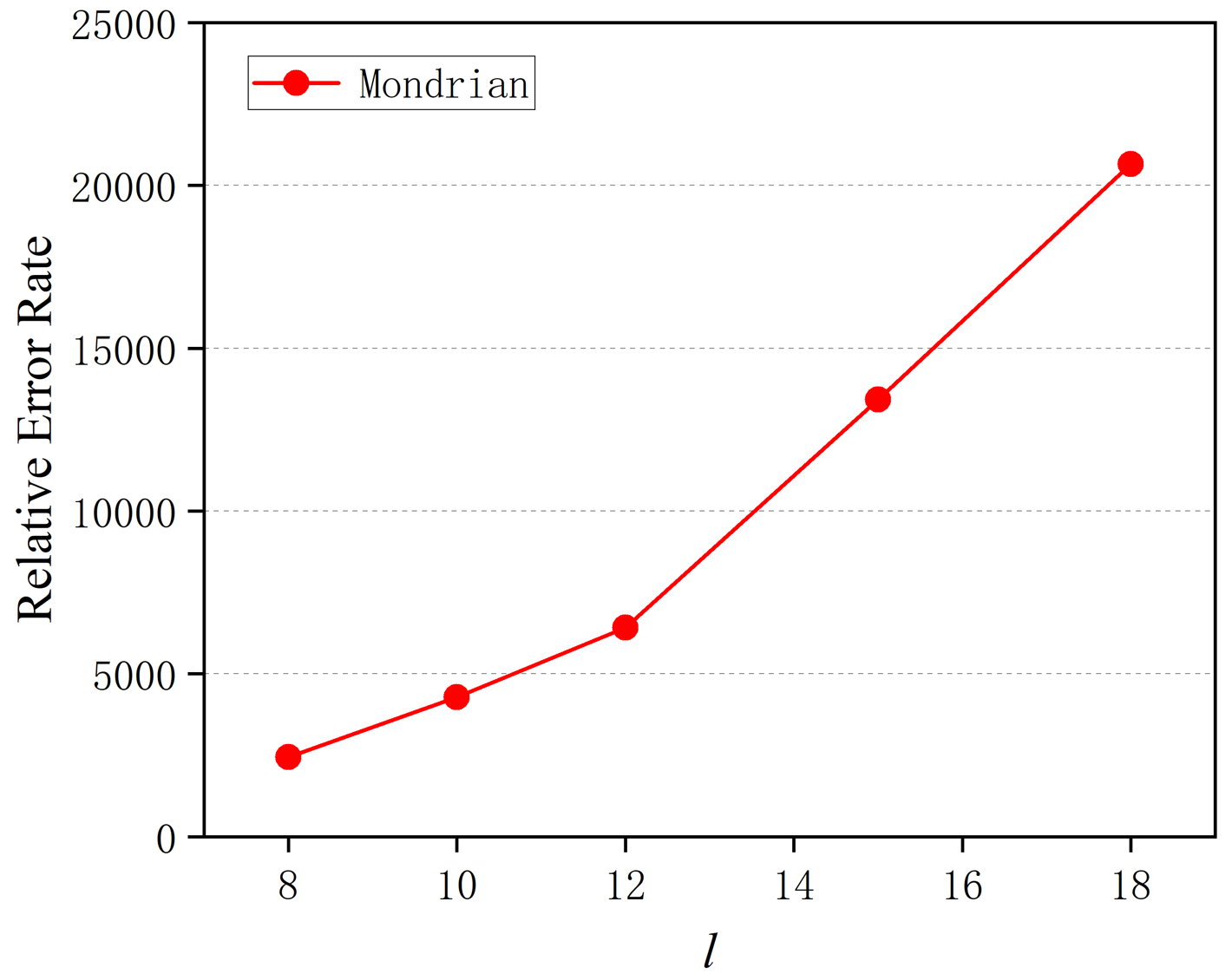}
		\label{fig_query_mondrian}}
	\subfigure[Anatomy]{\includegraphics[height=1.4in]{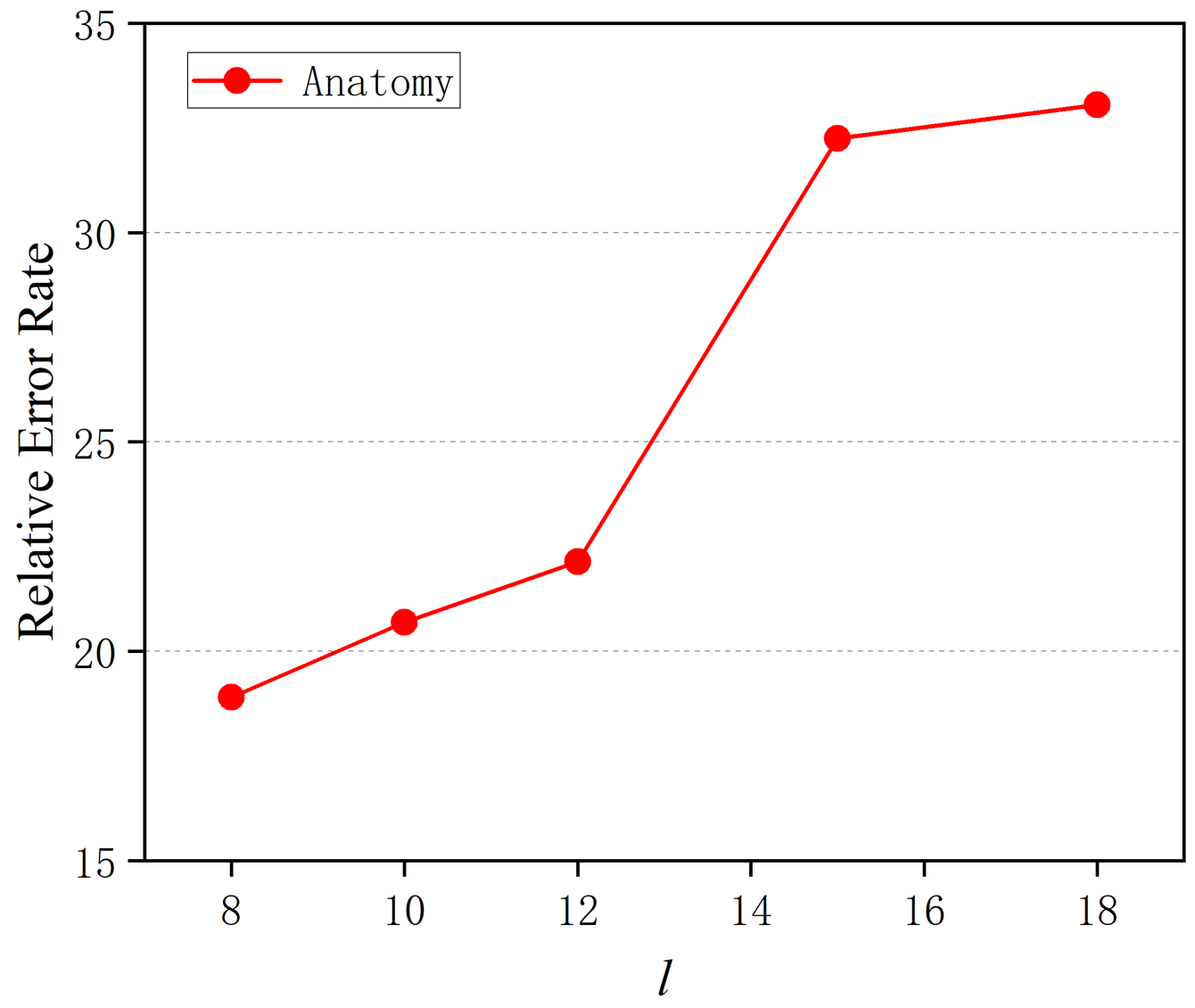}
		\label{fig_query_anatomy}}
	\caption{Average values of relative error rates.}
	\label{fig_query}
\end{figure*}

\begin{figure*}[!t]
	\centering
	\subfigure[MuCo]{\includegraphics[height=1.4in]{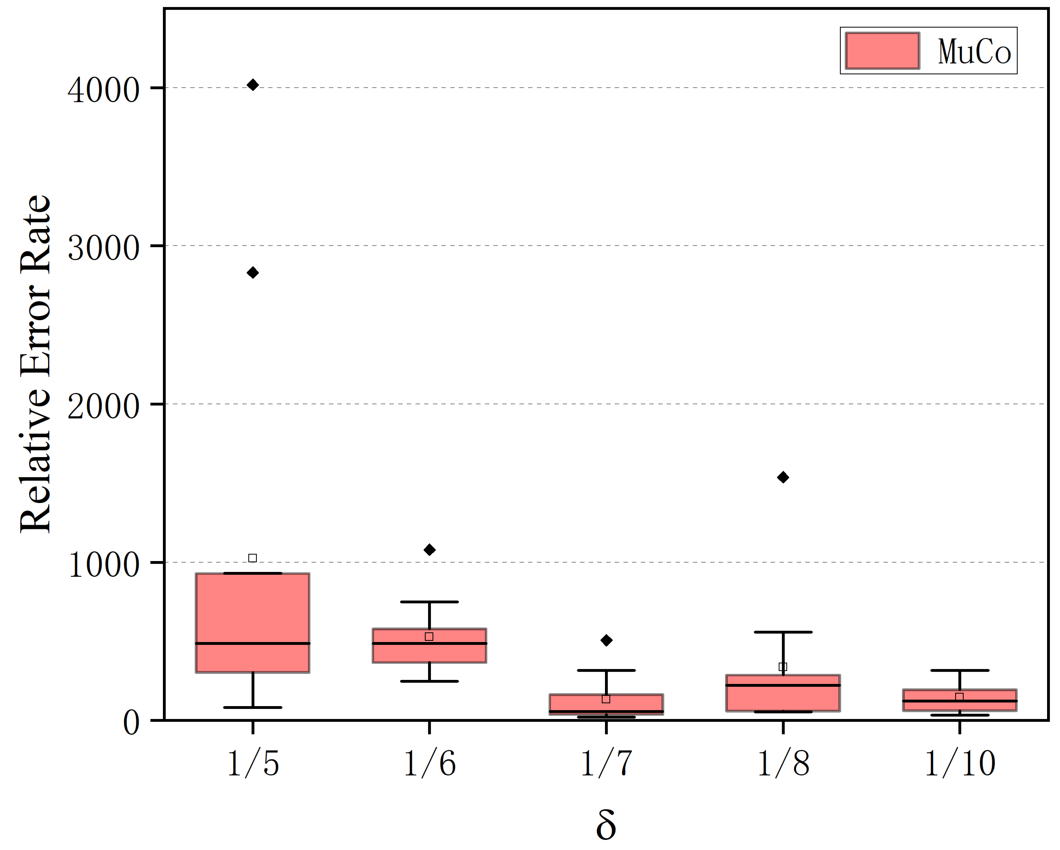}
		\label{fig_var_mutual}}
	\subfigure[Mondrian]{\includegraphics[height=1.4in]{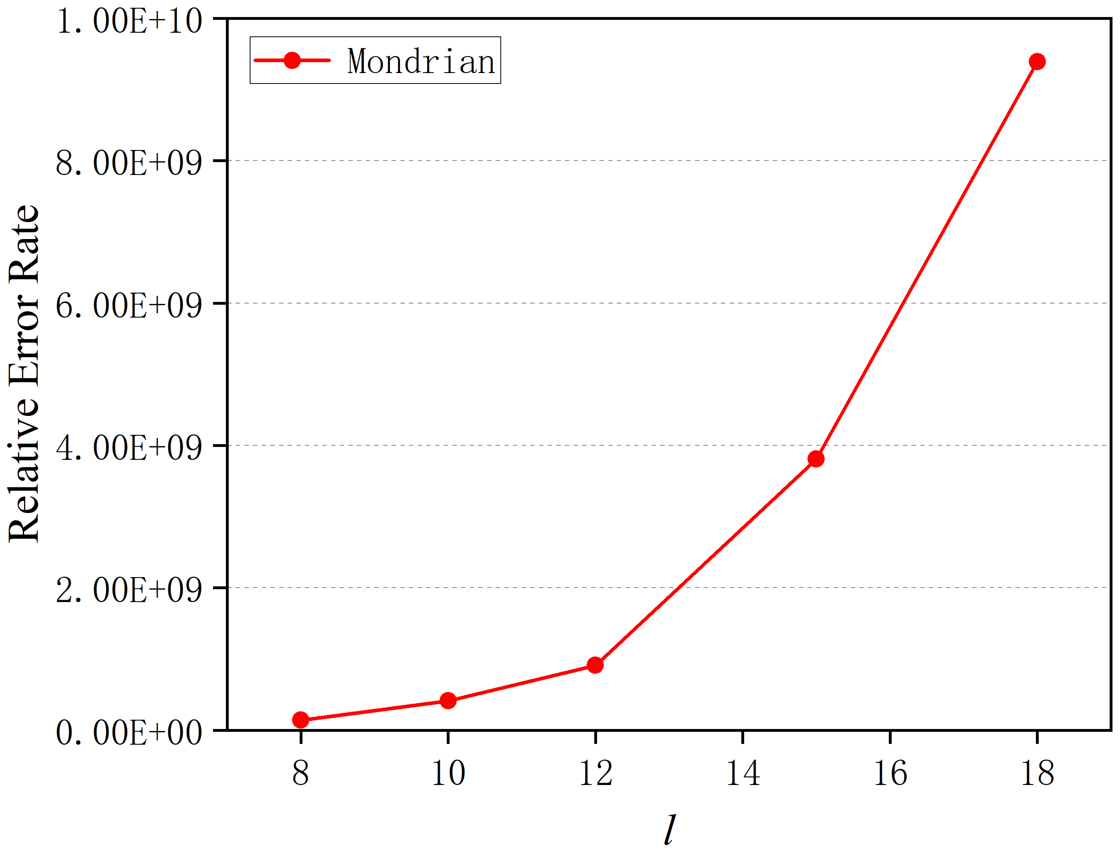}
		\label{fig_var_mondrian}}
	\subfigure[Anatomy]{\includegraphics[height=1.4in]{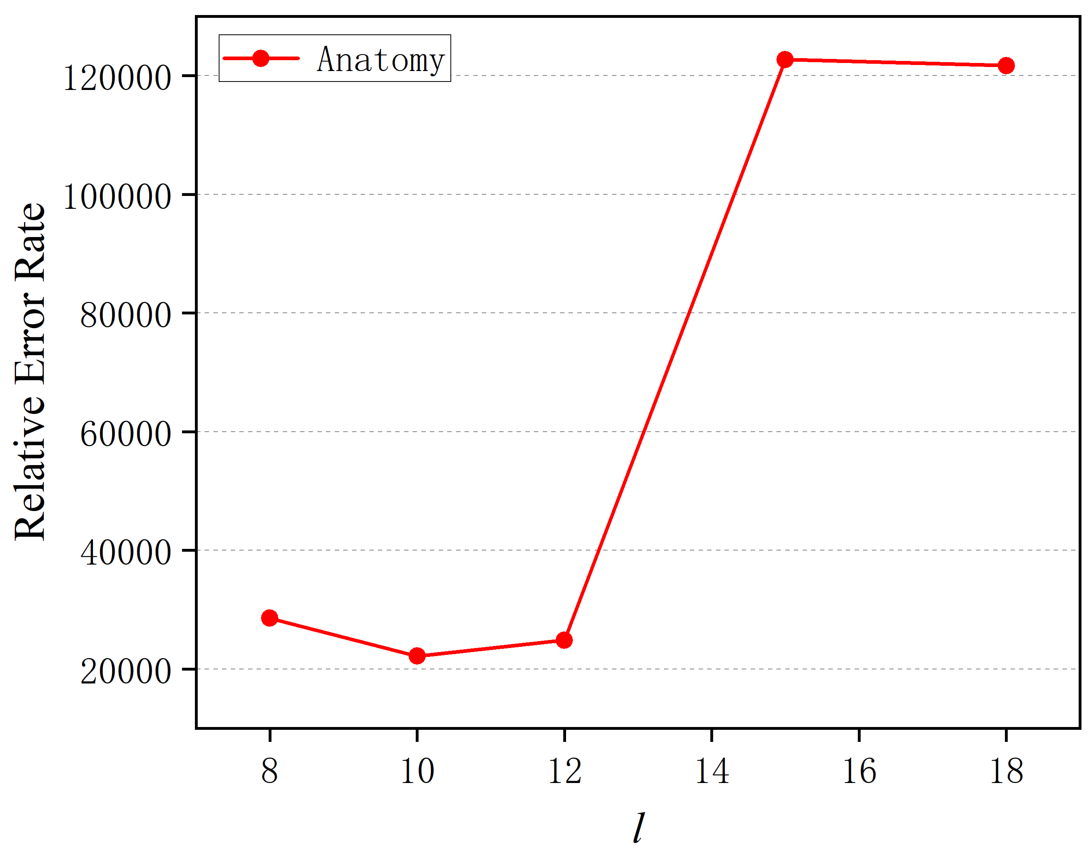}
		\label{fig_var_anatomy}}
	\caption{Variances of relative error rates.}
	\label{fig_var}
\end{figure*}

\begin{figure}[!t]
	\centering
	\subfigure[Information loss]{
		\begin{minipage}[t]{0.45\linewidth}
			\centering
			\includegraphics[height=1.2in]{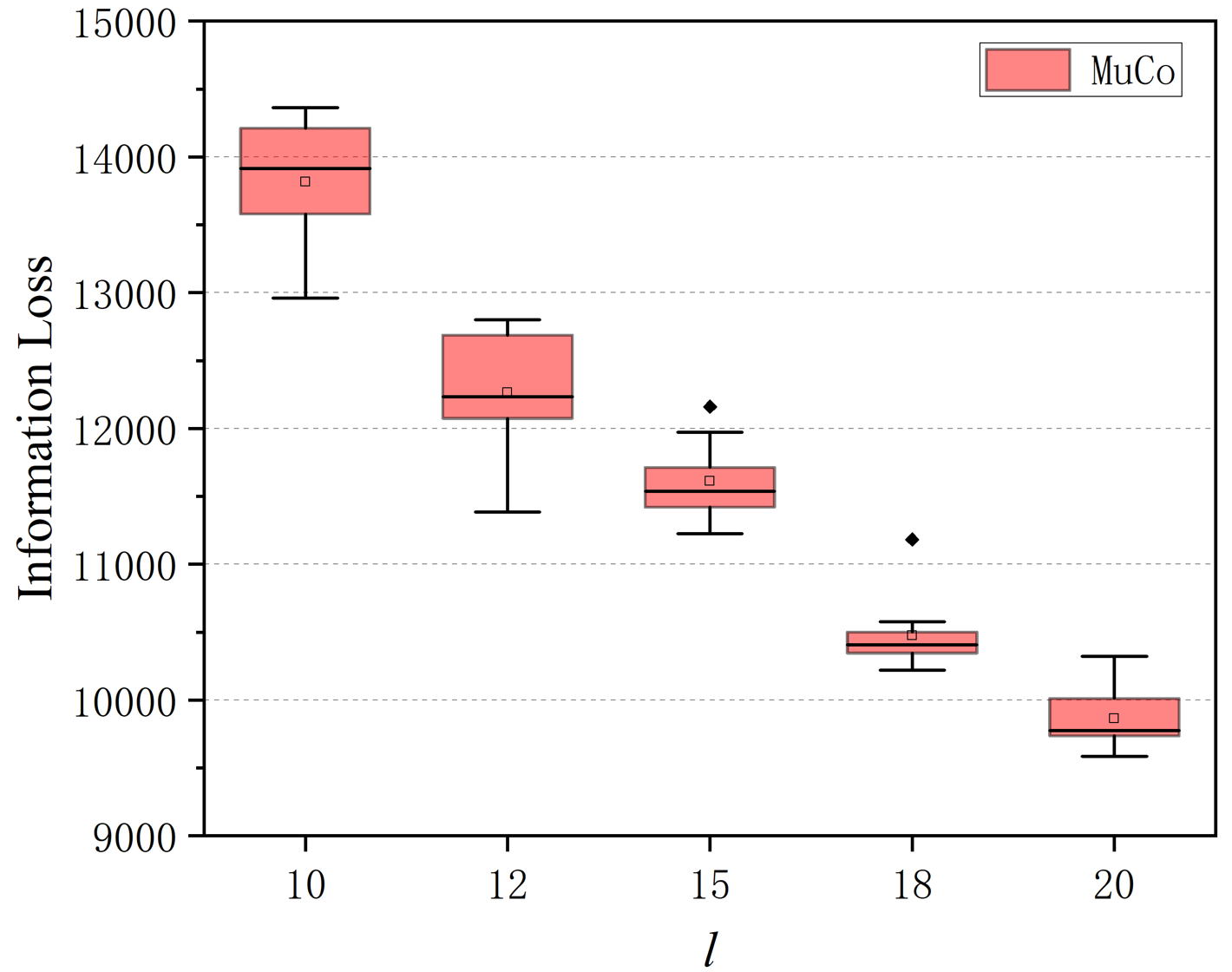}
			\label{fig_eff_loss}
	\end{minipage}}
	\subfigure[Relative error rate]{
		\begin{minipage}[t]{0.45\linewidth}
			\centering
			\includegraphics[height=1.2in]{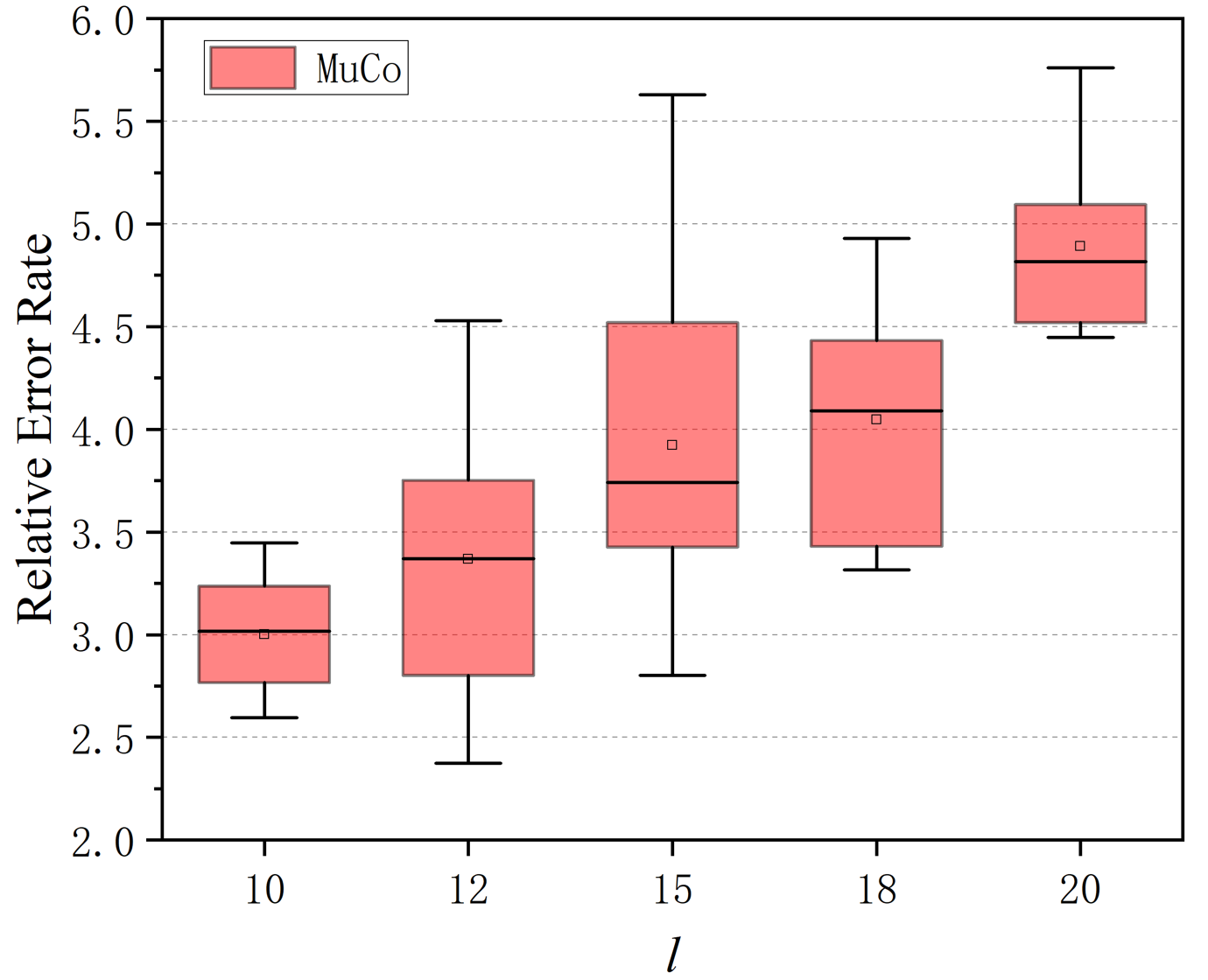}
			\label{fig_eff_query}
	\end{minipage}}
	\caption{Influence of attribute protection.}
	\label{fig_eff}
\end{figure}

We observe that the results of MuCo are much better than that of Mondrian and Anatomy. The primary reason is that MuCo retains the most distributions of the original QI values and the results of queries are specific records rather than groups. Consequently, the accuracy of query answering of MuCo is much better and more stable than that of Mondrian and Anatomy. Besides, since the results of queries for MuCo are specific records rather than groups, the relative error rate of MuCo does not increase steadily with the growth of $\delta$ but fluctuates depending on specific query conditions. Therefore, differing from Mondrian and Anatomy, increasing the level of protection of MuCo has little influence on the query results. In conclusion, MuCo can achieve the same level of protection as generalization does but with less information loss and more accurate query results. Note that, since we use the sum of salary for comparison (the range of salary is from 4 to 718,000), the relative error rates of Mondrian are much larger than some existing works.

\subsection{Cost of Attribute Protection}
\label{sec_effect}
In this experiment, we evaluate the influence of attribute protection on information loss and query answering. We set $l$ to 10, 12, 15, 18, and 20, respectively, assign $\delta$ to $\frac{1}{6}$, and use the same configuration as previous experiments. The results are shown in Figure \ref{fig_eff}.

Results from Figure \ref{fig_eff} show that the increase of $l$ lowers the information loss but raises the relative error rate. It is mainly because the number of tuples in each group increases with the growth of $l$. On the one hand, in random output tables, the probabilities that tuples have to cover on the QI values of long distance reduce significantly, and at the same time, the range of random output values (i.e., the column values) also becomes larger. Besides, the influence of $l$ reduce little information utility, such that MuCo avoids the problem of over-protection in generalization.

\section{Conclusion}
\label{sec_conclu}
In this work, we propose a novel technique, called the Mutual Cover (MuCo), to protect the privacy for microdata publication. The rationale is to make similar records to cover for each other at the minimal cost by perturbing the original QI values according to the random output tables. In this way, MuCo can achieve great protection performance, and the anonymization process is hidden for the adversary. Furthermore, MuCo preserves more information utility than generalization because the distributions of the original QI values are preserved as much as possible and the results of query statement are specific matching tuples rather than groups. Additionally, MuCo avoids the problem of over-protection for identities. The experiments illustrate that MuCo provides impressive privacy protection, little information loss, and accurate query answering.

\bibliographystyle{cas-model2-names}
\bibliography{references}

\end{document}